\documentclass[twoside, twocolumn, journal]{IEEEtran}
\IEEEoverridecommandlockouts
\usepackage{hyperref}
\usepackage{xcolor}
\usepackage{textcase}
\usepackage{setspace}
\usepackage{lmodern}
\usepackage{array}
\usepackage{tikz}
\usepackage{cite}
\usepackage[section]{placeins}
\usepackage{bm}
\usepackage{subcaption}
\usepackage{float}
\usepackage{amsmath}
\usepackage{graphicx}
\usepackage{epstopdf}
\usepackage{amssymb}
\usepackage{amssymb,amsmath}
\usepackage{amsgen,amsfonts,amsbsy,amsthm}
\usepackage{latexsym}
\usepackage{times}
\usepackage{helvet}
\usepackage[utf8]{inputenc}
\makeatletter
\def\BState{\State\hskip-\ALG@thistlm}
\makeatother
\DeclareMathOperator*{\argmin}{arg\,min}

\newcounter{defcounter}
\newcommand{\inprod}[2]{\left\langle #1,#2 \right\rangle}
\newcommand{\bvec}[1]{\boldsymbol{#1}}
\newcommand{\real}{\mathbb{R}}
\newcommand{\support}{\texttt{supp}}
\newcommand{\proj}[1]{\mathbf{P}_{#1}}
\newcommand{\dualproj}[1]{\mathbf{P}_{#1}^\perp}
\newcommand{\opnorm}[1]{\left\|#1\right\|}
\newcommand{\abs}[1]{\left|#1\right|}
\newcommand{\diag}[1]{\texttt{diag}\left(#1\right)}

\newtheorem{lem}{Lemma}[section]
\newtheorem{thm}{Theorem}[section]

\newtheorem{define}{Definition}[section]
\title{Modified Hard Thresholding Pursuit with Regularization Assisted Support Identification}
\begin{document}
		\setlength\abovedisplayskip{0pt}
	\setlength\belowdisplayskip{0pt}
\author{Samrat Mukhopadhyay$^1$, {\sl Student Member, IEEE}, and Mrityunjoy
Chakraborty$^2$, {\sl Senior Member, IEEE}



\thanks{The authors are with the department of Electronics and Electrical Communication
Engineering, Indian Institute of Technology, Kharagpur, INDIA(email : $^1$samratphysics@gmail.com, $^2$mrityun@ece.iitkgp.ernet.in).
}}
    %
\maketitle
\begin{abstract}
Hard thresholding pursuit (HTP) is a recently proposed iterative sparse recovery algorithm which is a result of combination of a support selection step from iterated hard thresholding (IHT) and an estimation step from the orthogonal matching pursuit (OMP). HTP has been seen to enjoy improved recovery guarantee along with enhanced speed of convergence. Much of the success of HTP can be attributed to its improved support selection capability due to the support selection step from IHT. In this paper, we propose a generalized HTP algorithm, called regularized HTP (RHTP), where the support selection step of HTP is replaced by a IHT-type support selection where the cost function is replaced by a regularized cost function, while the estimation step continues to use the least squares function. With decomposable regularizer, satisfying certain regularity conditions, the RHTP algorithm is shown to produce a sequence dynamically equivalent to a sequence evolving according to a HTP-like evolution, where the identification stage has a gradient premultiplied with a time-varying diagonal matrix. RHTP is also proven, both theoretically, and numerically, to enjoy faster convergence vis-a-vis HTP with both noiseless and noisy measurement vectors. 
\end{abstract}
\begin{IEEEkeywords}
Compressed sensing, Hard Thresholding Pursuit (HTP), Regularization.
\end{IEEEkeywords}
\section{Introduction}
\label{sec:introduction}
Compressed sensing is a recently developed paradigm which considers solving inverse problems by recovering sparse large unknwon vectors from small number of measurement vectors, linearly obtained form the unknown vector. Specifically, in compressed sensing, a system of linear equations, $\bvec{y}=\bvec{\Phi x}$ is provided, where $\bvec{y}$ is a $m$ dimensional measurement vector, $\bvec{\Phi}$ is a $m\times m$ dimensional measurement matrix and $\bvec{x}$ is the unknown $n$ dimensional vector. The goal is to estimate $\bvec{x}$. In compressed sensing, $m<<n$, which makes it an undetermined set of linear equations having infinitely many solutions. However, if $\bvec{x}$ is assumed to be $K$-sparse,i.e., if $K$ of the coordinates of $\bvec{x}$ are zero, such that $K<m/2$, then it is possible to recover $\bvec{x}$ exactly with $\mathcal{O}(K\ln(n/K))$ number of measurements~\cite{candes2006robust}. Candes and Tao~\cite{candes_decoding_2005} first showed that such a vector $\bvec{x}$ can be found exactly by solving the following $l_1-$optimization problem \begin{align*}
\ & \min_{\bvec{z}\in \real^n}\opnorm{\bvec{z}}_1\\
\ \mathrm{s.t.} & \ \bvec{y} = \bvec{\Phi z}
\end{align*}
With little effort, the above problem can be cast as a linear programming (LP) problem and can be solved using  standard LP routines. However, if the number of unknowns $n$ is large, this approach becomes computationally expensive~\cite{donoho2009message}. An alternate path to solution is provided by greedy algorithms, which leverage the knowledge of sparsity $K$ of $\bvec{x}$, and iteratively try to estimate $\bvec{x}$. Orthogonal matching pursuit (OMP)~\cite{pati1993orthogonal,tropp2004greed} and iterative hard thresholding (IHT)~\cite{daubechies2004iterative,herrity2006sparse} are two such classical algorithms with distinct estimation strategies. OMP iteratively searches for columns of the measurement matrix, one per each iteration, which are most suitable to correspond to the support of the unknown vector. Whereas, IHT minimizes a quadratic approximation of the least square cost function along with the $K$-sparsity constraint~\cite{blumensath2009iterative}. Both these algorithms are popular for their simplicity and ease of implementation. However, OMP has weaker theoretical recovery guarantee than IHT and it requires larger computational costs due to a projection step, which in turn ensures convergence in finite number of iterations. On the other hand, IHT is slower and may require infinite number of iterations while requiring much smaller computational cost.

Several other algorithms, like compressive sampling matching pursuit(CoSaMP)~\cite{needell2009cosamp}, subspace pursuit(SP)~\cite{dai2009subspace}, hard thresholding pursuit~\cite{foucart2011hard} were proposed to partially overcome the drawbacks of OMP and IHT to result in better convergence behaviour. Of these, HTP is particularly impressive as it is derived by taking the best of both OMP and IHT, and has been shown to have superior recovery performance to CoSaMP and SP~\cite{foucart2011hard} and requires much smaller, finite number of iterations for recovery as well~\cite{bouchot2016hard}. The key to the improved behaviour of HTP is attributed to the fact that at each iteration, it uses a two stage strategy, where in the first \emph{identification} stage it identifies a support with $K$ indices and then uses it in the following \emph{estimation} stage to estimate a $K$-sparse vector that minimizes the least squares cost function with the solution constrained to be on that support. Since the estimation stage ensures that the estimate is one of the $\binom{n}{K}$ many solutions corresponding to the $\binom{n}{K}$ possible supports of size $K$, it plays the principal role in ensuring that the algorithm converges in finite number of steps. However, the fast recovery performance of HTP is also highly dependent upon the quality of the support selected by the identification step, as a sequence of ``bad'' support selection might lead to delayed recovery.

In this paper we propose a generalization of HTP where  where we modify the identification stage by adding an extra gradient term of a regularization function to the gradient of the objective function. We call this method, the regularized HTP (RHTP). Note that if the regularizer function is takes as $0$, it becomes the HTP algorithm. This modification was partly motivated by observing the equivalency of the following two problems \begin{align*}
\ & \min_{\bvec{x}\in \real^n} \opnorm{\bvec{y} - \bvec{\Phi x}}_2^2\\
\ & \mathrm{s.t.}\opnorm{\bvec{x}}_0\le K,\ J(\bvec{x})\le R,
\end{align*} and \begin{align*}
\ & \min_{\bvec{x}\in \real^n} \opnorm{\bvec{y} - \bvec{\Phi x}}_2^2 + \sum_{j=1}^n \gamma_j g_j({x}_j)\\
\ & \mathrm{s.t.} \opnorm{\bvec{x}}_0\le K,
\end{align*} as the latter is formed from the former by adding a Lagrange multiplier to the cost function corresponding to the second constraint, where we have assumed that the function $J:\real^n\to \real$, is decomposable, i.e., $J(\bvec{x}) = \sum_{j=1}^n g_j(x_j)$, and that the Lagrange multipliers are $\{\gamma_1,\cdots,\gamma_n\}$. The presence of a regularizer generally corresponds to availability of prior information of the unknown vector to be estimated, and often improves the accuracy of the estimator. For example, the estimator produced by the elastic net (EN) regularizer by Zou and Hastie~\cite{zou2005regularization}, which uses a convex combination of the Ridge($l_2^2$)-regularizer and the $l_1$-regularizer, posses a certain ``grouping'' effect ~\cite[Theorem 1]{zou2005regularization}, where EN produces nearly equal estimates for a group of the nonzero entries in the support of the true vector if the corresponding columns of the measurement matrix are highly correlated. However, the $l_1$-minimizer alone has been seen to select only one of these columns and discard the others. Thus, using gradient corresponding to a regularized cost function in the identification step, it is hoped that the modified support selection step of the HTP algorithm will be useful in ``better'' convergence of the algorithm in some way. 

We provide a convergence analysis of the RHTP algorithm to find conditions on the measurement matrix as well the regularizer function and certain parameters of the algorithm, under which the algorithm converges. In the process of doing so it is revealed that if the regularizer function is \emph{decomposable} and if the decomposed functions have certain regularity properties, the sequence generated by the RHTP algorithm is \emph{topologically conjugate}~\cite{devaney2008introduction}, i.e. is dynamically equivalent, to a transformed sequence which has a HTP-like evolution with the gradient of the identification step premultiplied by a time varying diagonal matrix. Using this alternate HTP-like evoltuion, we show that if the common paramteres of HTP and RHTP are kept the same, the latter enjoys faster convergence speed compared to HTP. We also analyze the number of iterations it takes for RHTP to recover the true support with both noiseless and noisy measurements, and show that RHTP requires less number of iterations compared to HTP to do so. Finally, we corroborate our theoretical findings by performing several numerical experiments.
\section{Preliminaries}
\label{sec:preliminaries}
We first define the notion of equivalence between two dynamical systems, related by mappings, which will be helpful in the analysis later. For this purpose we use the notion of \emph{topological conjugacy} between two maps generating two dynamics~(See Devaney~\cite[\S 1.7]{devaney2008introduction}). We borrow the definition for topological conjugacy from Devaney~\cite{devaney2008introduction} and state it as below:
\begin{define}[Topological Conjugacy]
\label{def:topological-conjugacy}
Let $A,B$ be two topological spaces. Two mappings $f_1: A\to A$, and $f_2: B\to B$ are topologically conjugate if there is a homeomorphism $h:A\to B$ such that $h\circ f_1 = f_2 \circ h$.
\end{define} 
For two topological spaces $A, B$ equipped with topologies $\mathcal{T}_A,\ \mathcal{T}_B$, a homeomorphism $h:A\to B$ is a mapping such that $h$ is a bijection, and both $h$ and $h^{-1}$ are continuous. Thus intuitively topological conjugacy between two mappings $f_1,\ f_2$ implies that same dynamics can be represented in two different forms using the maps $f_1, f_2$. The importance of topological conjugacy is in the fact that if the maps $f_1,f_2$ are topologically conjugate, the properties of the sequence $\{f_1^n(\bvec{x})\}_{n\ge 0}$ can be studied by studying only the properties of the sequence $\{f_2^{n}(h(\bvec{x}))\}_{n\ge 0}$.

We use the following notations in the sequel:\begin{enumerate}
\item For any subset $S\subseteq \{1,\cdots,n\}$, $S^C$ denotes the complement of the subset $S$.
\item The superscript `$t$' is used to denote the transpose of a vector or matrix.
\item $\Omega_{n}^K$ denotes the collection of all subsets of size $K$ of the set $\{1,2,\cdots,\ n\}$.
\item $f:\real^n\to \real$ is such that $f(\bvec{x}) = \frac{1}{2}\opnorm{\bvec{y} - \bvec{\Phi x}}_2^2$.
\item $\psi_j:\real\to \real,\ j=1,2,\cdots,n$ are functions defined as $\psi_j(x) = x - \gamma_j g_j'(x),\ \forall x\in \real$.
\item The mapping $\bvec{\Psi}:\real^n\to \real^n$ is defined as $\bvec{\Psi}(\bvec{x}) = [\psi_1(x_1)\cdots\ \psi_n(x_n)]^t$, and the corresponding inverse is defined as $\bvec{\Psi}^{-1}(\bvec{x}) = [\psi_1^{-1}(x_1)\cdots\ \psi_n^{-1}(x_n)]^t$.
\item $\forall \bvec{x}\in \real^n$, $\bvec{D\Psi}$ and $\bvec{D\Psi}^{-1}(\bvec{x})$ denote the Jacobians of the transformations $\bvec{\Psi}$ and $\bvec{\Psi}^{-1}$, respectively, at the point $\bvec{x}$, i.e., $\left[\bvec{D\Psi}(\bvec{x})\right]_{ij} = \frac{\partial [\bvec{\Psi}(\bvec{x})]_i}{\partial x_j}$ and $\left[\bvec{D\Psi}^{-1}(\bvec{x})\right]_{ij} = \frac{\partial [\bvec{\Psi}^{-1}(\bvec{x})]_i}{\partial x_j},\ \forall i,j = 1,2,\cdots,\ n$.
\item $\bvec{M}:\real^n\to \real^{n\times n}$ is defined as $\bvec{M}(\bvec{x}) = \diag{\psi_1'(\psi_1^{-1}(x_1)),\cdots,\ \psi_n'(\psi^{-1}_n(x_n))},$ and $\bvec{M}^{-1}(\bvec{x}) = \left(\bvec{M}(\bvec{x})\right)^{-1}$. Also, for any subset $T\subseteq \{1,2,\cdots,n\}$ such that $T=\{i_1,\cdots,\ i_l\}$, $\bvec{M}_T(\bvec{x}): = \diag{\psi_{i_1}'(\psi^{-1}_{i_1}(x_{i_1})),\cdots,\psi_{i_l}'(\psi^{-1}_{i_l}(x_{i_l}))}$, and $\bvec{M}^{-1}_T(\bvec{x}) = (\bvec{M}_T(\bvec{x}))^{-1}$.
\item $w:\real^n\to \real$ is defined such that $w(\bvec{x}) = f(\bvec{\Psi}^{-1}(\bvec{x})),\ \forall \bvec{x}\in \real^n$.
\item $\chi_1:\real^n\to \real^n$ is defined such that \begin{align*}
\chi_1(\bvec{x}) = \argmin_{\bvec{z}:\support(\bvec{z}) \subseteq S_1(\bvec{x})}f(\bvec{z}),
\end{align*} 
where $S_1:\real^n\to \Omega_n^K$ is defined such that $ S_1(\bvec{x}) = \support\left(H_K\left(\bvec{x} - \mu \nabla f(\bvec{x}) - \bvec{\Gamma} \nabla J(\bvec{x})\right)\right)$.
\item $\chi_2:\real^n\to \real^N$ is defined such that \begin{align*}
\chi_2(\bvec{x}) = \argmin_{\bvec{z}:\support(\bvec{z}) \subseteq S_2(\bvec{x})}w(\bvec{z}),
\end{align*}  where $S_2:\real^n\to \Omega_n^K$ is defined as $ S_2(\bvec{x}) = \support\left(H_K\left(\bvec{x} - \mu\bvec{M}(\bvec{x}) \nabla w(\bvec{x})\right)\right)$.
\item $\bvec{r}(\bvec{x})\in \real_+^n$ is the nondecreasing arrangement of a vector $\bvec{x}\in \real^n$, i.e., $r_1(\bvec{x})\ge r_2(\bvec{x})\ge \cdots\ge r_N(\bvec{x})\ge 0$, and there exists a permutation $\pi$ of $\{1,2,\cdots,\ n\}$ such that $r_{j}(\bvec{x}) = \abs{x_{\pi(j)}},\ \forall j=1,2,\cdots,\ n$.
\item For any positive integer $s$, $\rho_{3s} = \frac{\sqrt{2}\mu'_{3s}}{\sqrt{1 - (\mu'_{2s})^2}}$, and $\tau_s = \frac{\sqrt{2}\mu \sqrt{(1+\delta_{s})}}{\sqrt{1 - (\mu'_{s})^2}} + \frac{\sqrt{1+\delta_{s}}}{1-\delta_{s}}$ where $\mu'_{s} = \left(1 - \frac{\mu(1-\delta_{s})}{1-l}\right)$.
\end{enumerate}

Furthermore, we require the matrix $\bvec{\Phi}$ to satisfy the restricted isometry property (RIP) which, informally, stands for the near unitariness of the columns of the matrix $\bvec{\Phi}$. Till its inception from the seminal paper by Candes and Tao~\cite{candes2006robust}, RIP has been used as a standard tool for the analysis of algorithms in the literature of compressed sensing, and has been used predominantly in the analysis of algorithms like HTP, CoSaMP and SP~\cite{foucart2011hard,needell2009cosamp,dai2009subspace}.
%
\begin{define}[Restricted Isometry Property]
\label{def:RIP-property}
For any integer $K>0$, the matrix $\bvec{\Phi}$ is said to satisfy the restricted isometry property (RIP) of order $K$ with restricted isometry constant (RIC) $\delta_K\in [0,1)$, if $\forall \bvec{x}\in \real^n,\ \mathrm{s.t.}\ \opnorm{\bvec{x}}_0\le K$, \begin{align}
\label{eq:RSS/SC-property}
(1-\delta)\opnorm{\bvec{x}}_2^2 & \le \opnorm{\bvec{\Phi x}}_2^2\le (1+\delta)\opnorm{\bvec{x}}_2^2,
\end{align}
for all $\delta\ge \delta_K$.
\end{define}
We will also require the following lemmas in the analysis:
\begin{lem}[Wielandt~\cite{horn1990matrix}]
\label{lem:wielandt-theorem}
For any integer $L>0$, let $\bvec{B}\in \real^{L\times L}$ be a real positive definite matrix with maximum and minimum eigenvalues $\lambda_M$, and $\lambda_m>0$, respectively. Then, $\forall \bvec{x}, \bvec{y}\in \real^L$, such that $\inprod{\bvec{x}}{\bvec{y}}=0$, the following is satisfied:\begin{align}
\label{eq:wielandt-inequality}
\abs{\bvec{x}^t\bvec{By}}^2 & \le \left(\frac{\lambda_M - \lambda_m}{\lambda_M + \lambda_m}\right)^2 (\bvec{x}^t\bvec{Bx})(\bvec{y}^t\bvec{By}).
\end{align}
Moreover, there exists orthonormal pair of vectors $\bvec{x},\ \bvec{y}\in \real^L$ such that the the above inequality is satisfied with equality.
\end{lem}
The following lemma will also be useful for our analysis. Similar result has also been obtained in Lemma~$3$ in Wen \emph{et al}~\cite{wen2017nearly} which used Lemma $2.1$ of Chang and Wu~\cite{chang2014improved} to prove their result. However, we will be using Lemma~\ref{lem:wielandt-theorem} due to Wielandt to prove the following lemma.
\begin{lem}
\label{lem:norm-of-projection-upper-bound}
If $S$ is a subset of $\mathcal{H}$, and $S\cap \Lambda=\emptyset$, and if $\bvec{y} = \bvec{\Phi}_\Lambda\bvec{x}_\Lambda$, and if the matrix $\bvec{\Phi}$ satisfies the RIP of order $\abs{S}+\abs{\Lambda}$ with RIC $\delta_{|S|+|\Lambda|}$, then, \begin{align}
\label{eq:norm-of-projection-upper-bound}
\opnorm{\proj{S}\bvec{y}}_2 \le \delta_{|S|+|\Lambda|}\opnorm{\bvec{y}}_2
\end{align}
\end{lem}
\begin{proof}
The proof is provided in Appendix~\ref{sec:appendix-proof-of-lemma-norm-of-projection-upper-bound}.
\end{proof}
\begin{lem}
\label{lem:inner-product-of-projections-upper-bound}
If $S$ is subset of $\mathcal{H}$ such that $S\cap \Lambda=\emptyset$, and if $\bvec{y} = \bvec{\Phi}_\Lambda\bvec{x}_\Lambda$, and if the matrix $\bvec{\Phi}$ satisfies the RIP of order $\abs{S}+\abs{\Lambda}+1$ with RIC $\delta_{|S|+|\Lambda|+ 1}$, then, $\forall i\notin S\cup \Lambda$,\begin{align}
\label{eq:inner-product-of-projections-upper-bound}
\abs{\bvec{\phi}^t_i\dualproj{S}\bvec{y}} \le \delta_{\abs{S}+\abs{\Lambda}+1}\opnorm{\dualproj{S}\bvec{\phi}_i}_2\opnorm{\bvec{y}}_2
\end{align}
\end{lem}
\begin{proof}
The proof is provided in Appendix~\ref{sec:appendix-proof-of-lemma-inner-product-of-projections-upper-bound}.
\end{proof}
\begin{lem}
\label{lem:homeomorphism-Psi}
The map $\bvec{\Psi}:\real^n \to \real^n$ is a homeomorphism.
\end{lem}
\begin{proof}
The proof is provided in Appendix~\ref{sec:appendix-proof-lemma-homemorphism-Psi}.
\end{proof}
\begin{lem}
\label{lem:jacobians-of-psi-and-psi-inverse}
For any $\bvec{x}\in \real^n$, the Jacobians $\bvec{D\Psi}(\bvec{x})$ and $\bvec{D\Psi}^{-1}(\bvec{x})$ satisfy $\bvec{D\Psi}(\bvec{x}) = \bvec{M}(\bvec{\Psi}(\bvec{x}))$ and $\bvec{D\Psi}^{-1}(\bvec{x}) = \bvec{M}^{-1}(\bvec{x})$, respectively.
\end{lem}
\begin{proof}
The proof is provided in Appendix~\ref{sec:appendix-proof-lemm-jacobians-of-psi-and-psi-inverse}.
\end{proof}
%
%
\section{Proposed algorithm}
\label{sec:algo-description}
The RHTP algorithm is described in ~\autoref{tab:mixed-function-htp}.
\begin{table}[tb!]
\caption{\textsc{Algorithm}: Mixed Function HTP (RHTP)}
\label{tab:mixed-function-htp}
\begin{tabular}{p{9cm}}
\centering
\hrulefill
\begin{description}
\item[\textbf{Input:}]\ sparsity level $K$; Initial estimates $\bvec{x}^0$; step sizes $\mu>0$; a diagonal regularization matrix $\bvec{\Gamma}$ with positive diagonal elements; maximum number of iterations $k_{\mathrm{it}}$;
\end{description}
\begin{enumerate}
\item[\textbf{While}]$\ (k<k_{\mathrm{it}}):$ 
\item $\hat{\bvec{x}}^{k+1} = H_K\left(\bvec{x}^k + \mu \bvec{\Phi}^t(\bvec{y} - \bvec{\Phi x}^k) - \bvec{\Gamma} \nabla J(\bvec{x}^k)\right)$
\item $\Lambda^{k+1} = \support\left({\hat{\bvec{x}}^{k+1}}\right)$
\item $\bvec{x}^{k+1} = \argmin_{\bvec{z}:\support{\left(\bvec{z}\right)} \subseteq \Lambda^{k+1}} \opnorm{\bvec{y} - \bvec{\Phi x}^k}_2$
\item $k=k+1$
\item[\noindent\textbf{End While}]
\end{enumerate}
\hrulefill
\end{tabular}
\end{table}
where the diagonal elements of the diagonal matrix $\bvec{\Gamma}$ are the constants $\gamma_1,\cdots,\gamma_n$, and the function $J(\bvec{x})$ is assumed to be decomposable , i.e., we assume that there are functions $g_j:\real^\to \real,\ j=1,\cdots,n$, such that $J(\bvec{x})=\sum_{j=1}^n g_j(x_j)$.We also assume that the functions $g_j:\real\to \real$ satisfy the following:\begin{enumerate}
\item $g_j\in \mathcal{C}^2$, i.e. $g_j$ is continuously double differentiable.
\item $g_j'(0) = 0$. 
\item $g_j''(x) < 1/\gamma_j,\ \forall x\in \real.$
\end{enumerate}

From the Table~\ref{tab:mixed-function-htp} it can be seen that the only difference between RHTP and HTP is that the former has a $-\bvec{\Gamma }\nabla J(\bvec{x}^k)$ inside the hard thresholding operator in the identification stage. This modification induces the effect of availability of some prior information about the unknown vector in terms of the regularizer $J$ and the factors $\{\gamma_j\}_{j=1}^n$, and thus help select a support which is closer to the true support of the unknown vector, which in turn accelerates the convergence of the algorithm.  
\section{Theoretical results}
\label{sec:theoretical-results}
\subsection{Convergence analysis}
\label{sec:convergence-analysis}
This section analyzes convergence of the RHTP algorithm. The principle strategy used here is to find a sequence $\{\bvec{z}^k\}_{k\ge 0}$ which is topologically conjugate to the dynamics of the sequence $\{\bvec{x}^k\}$, and thus can be analyzed instead of the sequence $\{\bvec{x}^k\}_{k\ge 0}$ to understand the properties of the convergence. 
\begin{thm}
\label{thm:equivalent-htp-dynamics}
If the matrix $\bvec{\Phi}$ satisfies RIP of order $K$, then the mappings $\chi_1$ and $\chi_2$ are topologically conjugate.
\end{thm}
\begin{proof}
Let us fix some $\bvec{x}\in \real^n$. We need to find a homeomorphism $h:\real^n\to \real^n$ such that $h(\chi_1(\bvec{x})) = \chi_2(h(\bvec{x}))$. We claim that $h$ can be chosen to be the mapping $\bvec{\Psi}$. 

First, that $\bvec{\Psi}$ is a homeomorphism, follows from Lemma~\ref{lem:homeomorphism-Psi}. We now show that $\bvec{\Psi}(\chi_1(\bvec{x})) = \chi_2(\bvec{\Psi}(\bvec{x}))$. First observe that, while obtaining $\chi_1(\bvec{x})$, in the intermediate step we calculate the support $S_1(\bvec{x})$, which can be written as \begin{align*}
S_1(\bvec{x}) & = \support\left(H_K\left(\bvec{x} - \mu \nabla f(\bvec{x}) - \bvec{\Gamma} \nabla J(\bvec{x})\right)\right)\\
\ & = \support\left(H_K\left(\bvec{\Psi}(\bvec{x}) - \mu \nabla f(\bvec{x})\right)\right)\\
\ & \stackrel{(a)}{=} \support\left(H_K\left(\bvec{\Psi}(\bvec{x}) - \mu \bvec{M}(\bvec{\Psi}(\bvec{x}))\nabla w(\bvec{\Psi}(\bvec{x}))\right)\right)\\
\ & = S_2(\bvec{\Psi}(\bvec{x})),
\end{align*} 
where the step $(a)$ follows from the observation that $\nabla f(\bvec{x}) = \bvec{M}(\bvec{\Psi}(\bvec{x}))\nabla w(\bvec{\Psi}(\bvec{x}))$. Thus, \begin{align*}
\chi_2(\bvec{\Psi}(\bvec{x})) & = \argmin_{\bvec{z}:\support(\bvec{z}) \subseteq S_2(\bvec{\Psi}(\bvec{x}))}w(\bvec{z})\\
\ & = \argmin_{\bvec{z}:\support(\bvec{z}) \subseteq S_1(\bvec{x})}f(\bvec{\Psi}^{-1}(\bvec{z}))
\end{align*}
Hence, \begin{align*}
f(\bvec{\Psi}^{-1}(\chi_2(\bvec{\Psi}(\bvec{x})))) & = \min_{\bvec{z}:\support(\bvec{z}) \subseteq S_1(\bvec{x})}f(\bvec{\Psi}^{-1}(\bvec{z}))\\
\ & = \min_{\bvec{u}:\support(\bvec{\Psi}(\bvec{u})) \subseteq S_1(\bvec{x})}f(\bvec{u}))\\
\ & \stackrel{(b)}{=} \min_{\bvec{u}:\support(\bvec{u}) \subseteq S_1(\bvec{x})}f(\bvec{u}))\\
\ & \stackrel{(c)}{=} f(\chi_1(\bvec{x})),
\end{align*} 
where step $(b)$ follows from the observation that $\support(\bvec{\Psi}(\bvec{u})) = \support(\bvec{u}) $, since $\psi_j(0) = 0$ for all $j=1,2,\cdots,n$, and step $(c)$ follows from the definition of $\chi_1$. Now, observe that, by definition, $\chi_1(\bvec{x})$ is $K$-sparse with support $S_1(\bvec{x})$ and $\bvec{\Psi}^{-1}(\chi_2(\bvec{\Psi}(\bvec{x})))$ is also $K$-sparse with support $S_2(\bvec{\Psi}(\bvec{x}))$, as $\chi_2(\bvec{\Psi}(\bvec{x}))$ is $K$-sparse, by definition, and $\bvec{\Psi}^{-1}(\chi_2(\bvec{\Psi}(\bvec{x})))$ and $\chi_2(\bvec{\Psi}(\bvec{x}))$ have the same support thanks to the relation $\psi_j(0) = 0,\ \forall j\in \{1,2,\cdots,n\}$. Also, we have proved that $S_1(\bvec{x}) = S_2(\bvec{\Psi}(\bvec{x}))$. Furthermore, as the matrix $\bvec{\Phi}$ satisfies RIP of order $K$, for any support $S$ of size $K$, the minimizer $\argmin_{\bvec{z}:\support(\bvec{z})\subseteq S}f(\bvec{z}) = \argmin_{\bvec{z}:\support(\bvec{z})\subseteq S} \frac{1}{2}\opnorm{\bvec{y} - \bvec{\Phi}_S\bvec{z}_S}_2^2$ is unique, as the matrix $\bvec{\Phi}_S^t\bvec{\Phi}_S$ has minimum eigenvalue lower bounded by $1-\delta_K>0$, and hence is invertible. Consequently, we have $\chi_1(\bvec{x})=\bvec{\Psi}^{-1}(\chi_2(\bvec{\Psi}(\bvec{x})))\implies \bvec{\Psi}(\chi_1(\bvec{x})) = \chi_2(\bvec{\Psi}(\bvec{x}))$. This establishes the topological conjugacy of $\chi_1$ with $\chi_2$.
\end{proof}
Theorem~\ref{thm:equivalent-htp-dynamics} gives us an alternate dynamics furnished by $\chi_2$ which we can analyze to understand the convergence behavior of the dynamics governed by $\chi_1$. In the sequel we will analyze the sequence $\{\bvec{z}^k\}_{k\ge 0}$ instead of the sequence $\{\bvec{x}^k\}_{k\ge 0}$ produced by the RHTP algorithm, where $\bvec{z}^k = \bvec{\Psi}(\bvec{x}^k)$.
To further proceed we denote by $\Lambda^k$ the support of $\bvec{x}^k,\ \forall n\ge 0$, and $T^k = \Lambda^{k+1}\setminus \Lambda^k,\ U^k = \Lambda^{k+1}\cap \Lambda^k$.
\begin{lem}
\label{lem:bounds-on-iterates-zn-xhatn}
If the matrix $\bvec{\Phi}$ satisfies RIP of order $2K+1$, then at any iteration $n(\ge 1)$, the iterates $\bvec{x}^k,\ \bvec{z}^k,$ and $\hat{\bvec{x}}^{k+1}$ are upper bounded in the following way:\begin{align}
\label{eq:xn-upper-bound}
\abs{x_i^k } & \le B_k\le B\ \forall i\in \Lambda^k, \\
\label{eq:zn-upper-bound}
\abs{z_i^k} & \le C_{n,i}B_k\le C_iB\ \forall i\in \Lambda^k,\\
\label{eq:xhatn-upper-bound}
\abs{\hat{x}_i^{k+1}} & \le E_{n,i}\le E_i,\ \forall i\in \Lambda^{k+1},
\end{align}
 where \begin{align*}
B_k = \left\{\begin{array}{ll}
\delta_{2K}B, & \mbox{if $\Lambda^k\cap \Lambda=\emptyset$, in noiseless scenario}\\
B, & \mbox{else},
\end{array}\right.
\end{align*}
and $B = \frac{\opnorm{\bvec{y}}_2}{\sqrt{1-\delta_K}}$.
\begin{align*}
C_{k,i} = \max_{u\in[-B_k,B_k]}(1-\gamma_i g_i''(u)),
\end{align*} 
and $C_i = \max_{u\in[-B,B]}(1-\gamma_i g_i''(u))$. \begin{align*}
E_{k,i} = \max\{B_kC_{k,i},\ D_{k+1}\},
\end{align*} where \begin{align*}
D_k = \left\{\begin{array}{ll}
\delta_{2K+1}D, & \mbox{if}\ \Lambda^{k}\cap \Lambda=\emptyset\ \mbox{in noiseless scenario}\\
D, & \mbox{else},
\end{array}\right.
\end{align*}
where $D = \mu\max_i\opnorm{\bvec{\phi}_i}_2\opnorm{\bvec{y}}_2$, and $E_i = \max\{BC_i,\ D\}.$
\end{lem}
\begin{proof}
The proof is supplied in Appendix~\ref{sec:appendix-proof-lemma-bounds-on-iterates-zn-xhatn}.
\end{proof}
Our first result on the convergence of the sequence $\{\bvec{z}^k\}_{k\ge 0}$ is stated as below:
\begin{thm}
\label{thm:convergence-zn}
Let the matrix $\bvec{\Phi}$ satisfy the RIP of order $2K$. Also, let the functions $\{g_i\}_{1\le i\le L}$ satisfy the following property: \begin{align*}
\mu (1+\delta_{2K}) < \frac{(1- L)^2}{1- l},
\end{align*}
where $l = \min_i\min_{u\in [\psi_i^{-1}(-E_i),\psi_i^{-1}(E_i)]}\gamma_ig_i''(u),\ L = \max_i\max_{u\in [\psi_i^{-1}(-E_i),\psi_i^{-1}(E_i)]}\gamma_ig_i''(u)$, with $E_i$ as defined in Lemma~\ref{lem:bounds-on-iterates-zn-xhatn}. Then the RHTP algorithm converges in finite number of iterations.
\end{thm}
\begin{proof}
Writing $\bvec{u}^k(\theta) = \bvec{z}^k + \theta (\hat{\bvec{x}}^{k+1} - \bvec{z}^k)$, we get \begin{align*}
\lefteqn{w(\hat{\bvec{x}}^{k+1}) -  w(\bvec{z}^k)} &  &\\
\ & = \int_0^1\inprod{\nabla w(\bvec{z}^{k} + \theta(\hat{\bvec{x}}^{k+1} - \bvec{z}^k))}{\hat{\bvec{x}}^{k+1} - \bvec{z}^k} d\theta\\
\ & = \int_0^1\inprod{\bvec{M}^{-1}(\bvec{u}^k(\theta))\nabla f(\bvec{\Psi}^{-1}(\bvec{u}^k(\theta)))}{\hat{\bvec{x}}^{k+1} - \bvec{z}^k} d\theta\\
\ & = \int_0^1\inprod{\bvec{M}^{-1}(\bvec{u}^k(\theta))\nabla f(\bvec{\Psi}^{-1}(\bvec{u}^k(\theta)))}{\hat{\bvec{x}}^{k+1} - \bvec{z}^k} d\theta\\
\ & = \int_0^1\inprod{\bvec{M}^{-1}(\bvec{u}^k(\theta))\left(\nabla f(\bvec{\Psi}^{-1}(\bvec{u}^k(\theta))) - \nabla f(\bvec{x}^k)\right)}{\hat{\bvec{x}}^{k+1} - \bvec{z}^k} d\theta\\
\ & +  \int_0^1\inprod{\bvec{M}^{-1}(\bvec{u}^k(\theta))\nabla f(\bvec{x}^k)}{\hat{\bvec{x}}^{k+1} - \bvec{z}^k} d\theta\\
\ & = I_1 + I_2.
\end{align*}
We will now find upper bounds on $I_1,\ I_2$. 

To find an upper bound of $I_1$, first note that $\nabla f(\bvec{\Psi}^{-1}(\bvec{u}^k(\theta))) - \nabla f(\bvec{x}^k) = $ $\bvec{\Phi}^t\bvec{\Phi}(\bvec{\Psi}^{-1}(\bvec{u}^k(\theta)) - \bvec{x}^k)$. Now, we expess $\bvec{\Psi}^{-1}(\bvec{u}^k(\theta)) - \bvec{x}^k $ as \begin{align*}
\lefteqn{\bvec{\Psi}^{-1}(\bvec{u}^k(\theta)) - \bvec{x}^k} & &\\
\ & = \bvec{\Psi}^{-1}(\bvec{u}^k(\theta)) - \bvec{\Psi}^{-1}(\bvec{z}^k)\\
\ & = \int_0^1 \frac{d\bvec{\Psi}^{-1}(\bvec{z}^k + \theta_1(\bvec{u}^k(\theta) - \bvec{z}^k))}{d\theta_1}d\theta_1\\
\ & = \int_0^1 \bvec{D\Psi}^{-1}(\bvec{z}^k + \theta_1 (\bvec{u}^k(\theta) - \bvec{z}^k))(\bvec{u}^k(\theta) - \bvec{z}^k) d\theta_1\\
\ & \stackrel{(d)}{=}\theta\int_0^1 \bvec{M}^{-1}(\bvec{u}^k(\theta\theta_1))(\hat{\bvec{x}}^{k+1} - \bvec{z}^k) d\theta_1
\end{align*}
where step $(d)$ uses Lemma~\ref{lem:jacobians-of-psi-and-psi-inverse} and the identity $\bvec{u}^k(\theta) - \bvec{z}^k = \theta(\hat{\bvec{x}}^{k+1} - \bvec{z}^k)$ along with the definition of $\bvec{u}^k(\theta)$. Consequently, $I_1$ can be expressed as: $I_1 = (\hat{\bvec{x}}^{k+1} - \bvec{z}^k)_{S^k}^t \bvec{A}(\hat{\bvec{x}}^{k+1} - \bvec{z}^k)_{S^k}$, where $S^k = \Lambda^{k+1}\cup \Lambda^k$, and $\bvec{A} = \int_0^1 \int_0^1 \theta \bvec{M}^{-1}_{S^k}(\bvec{u}^k(\theta\theta_1))\bvec{\Phi}_{S^k}^t\bvec{\Phi}_{S^k}\bvec{M}_{S^k}^{-1}(\bvec{u}^k(\theta))d\theta_1 d\theta$. Thus, $I_1\le \lambda_{\max}\left(\bvec{A}\right)\opnorm{\hat{\bvec{x}}_{k+1} - \bvec{z}^k}_2^2$. Since the maximum eigenvalue is a convex function over the set of matrices, we have, due to Jensen's inequality, $ \lambda_{\max}\left(\bvec{A}\right) = \lambda_{\max}\left(\int_0^1 \int_0^1 \theta \bvec{M}^{-1}_{S^k}(\bvec{u}^k(\theta\theta_1))\bvec{\Phi}_{S^k}^t\bvec{\Phi}_{S^k}\bvec{M}_{S^k}^{-1}(\bvec{u}^k(\theta))d\theta_1 d\theta\right)\le \int_0^1 \int_0^1\lambda_{\max}\left(\theta \bvec{M}^{-1}_{S^k}(\bvec{u}^k(\theta\theta_1))\bvec{\Phi}_{S^k}^t\bvec{\Phi}_{S^k}\bvec{M}_{S^k}^{-1}(\bvec{u}^k(\theta))\right)d\theta_1 d\theta =  \int_0^1 \int_0^1\theta\lambda_{\max}\left(\bvec{M}_{S^k}^{-1}(\bvec{u}^k(\theta))\bvec{M}^{-1}_{S^k}(\bvec{u}^k(\theta\theta_1))\bvec{\Phi}_{S^k}^t\bvec{\Phi}_{S^k}\right)d\theta_1 d\theta.$ Now, for given $\theta, \theta_1$, $\lambda_{\max}\left(\bvec{M}_{S^k}^{-1}(\bvec{u}^k(\theta))\bvec{M}^{-1}_{S^k}(\bvec{u}^k(\theta\theta_1))\bvec{\Phi}_{S^k}^t\bvec{\Phi}_{S^k}\right) \le \lambda_{\max}\left(\bvec{M}_{S^k}^{-1}(\bvec{u}^k(\theta))\right)\lambda_{\max}\left(\bvec{M}_{S^k}^{-1}(\bvec{u}^k(\theta\theta_1))\right)\lambda_{\max}\left(\bvec{\Phi}_{S^k}^t\bvec{\Phi}_{S^k})\right)\le \max_{i\in S^k} M_{ii}^{-1}(\bvec{u}^k(\theta))\max_{i\in S^k} M_{ii}^{-1}(\bvec{u}^k(\theta\theta_1))(1+\delta_{2K})$. Now, to find an upper bound on $M_{ii}^{-1}(\bvec{u}^k(\theta))$, for any $\theta\in [0,1]$, we note that $M_{ii}^{-1}(\bvec{u}^k(\theta)) = \frac{1}{1 - \gamma_i g_i''(\psi^{-1}(u_i^k(\theta)))}$. Since, $u_i^k(\theta) = z_i^k+\theta(\hat{x}^{k+1}_i - z_i^k)$, using Lemma~\ref{lem:bounds-on-iterates-zn-xhatn}, we obtain, $\abs{u_i^k(\theta)}\le (1-\theta)C_{k,i}B_k + \theta E_{k,i}\le E_{k,i}$, which implies that $M_{ii}^{-1}(\bvec{u}^k(\theta)) = \frac{1}{1 - \gamma_i g_i''(\psi_i^{-1}(u_i^k(\theta)))}\le \frac{1}{1 - \gamma_i\max_{u \in [\psi_i^{-1}(E_{k,i}),\ \psi_i^{-1}(E_{k,i})]}g_i''(u)}=:F_{k,i}$. Thus, we obtain that $\forall \theta,\theta_1\in [0,1],\ \lambda_{\max}\left(\bvec{M}_{S^k}^{-1}(\bvec{u}^k(\theta))\bvec{M}^{-1}_{S^k}(\bvec{u}^k(\theta\theta_1))\bvec{\Phi}_{S^k}^t\bvec{\Phi}_{S^k}\right) \le F_k^2(1+\delta_{2K})$, where $F_k = \max_{i}F_{k,i}$. hence, $I_1\le \frac{F_k^2(1+\delta_{2K})\opnorm{\hat{\bvec{x}}^{k+1} - \bvec{z}^k}_2^2}{2}.$

To find an upper bound on $I_2$, we proceed as below:\begin{align*}
\lefteqn{\int_0^1\inprod{\bvec{M}^{-1}(\bvec{u}^k(\theta))\nabla f(\bvec{x}^k)}{\hat{\bvec{x}}^{k+1} - \bvec{z}^k} d\theta} & &\\
\ & \stackrel{(e)}{=}\int_0^1\inprod{\bvec{M}^{-1}(\bvec{u}^k(\theta))_{T^k}\nabla_{T^k} f(\bvec{x}^k)}{(\hat{\bvec{x}}^{k+1} - \bvec{z}^k)_{T^k}} d\theta,
\end{align*}
where step $(e)$ is due to the fact that $\nabla_{\Lambda^{k}}f(\bvec{x}^k) = \bvec{0}^{\Lambda^k}$. To further proceed, we utilize the structure of the vector $\hat{\bvec{x}}^{k+1}$, i.e., that $\hat{\bvec{x}}^{k+1}_{T^k} = -\mu \nabla_{T^k}f(\bvec{x}^k)$, and $\hat{\bvec{x}}^{k+1}_{U^k} = \bvec{z}^k_{U^k}$. Consequently, $\nabla_{T^k}f(\bvec{x}^k) = -\frac{1}{\mu}(\hat{\bvec{x}}^{k+1}_{T^k} - \bvec{x}^k_{T^k})$. Thus, we obtain, \begin{align*}
\lefteqn{\int_0^1\inprod{\bvec{M}^{-1}(\bvec{u}^k(\theta))\nabla f(\bvec{x}^k)}{\hat{\bvec{x}}^{k+1} - \bvec{z}^k} d\theta} & &\\
\ & = -\frac{1}{\mu}\int_0^1 \opnorm{\bvec{M}^{-1/2}(\bvec{u}^k(\theta))(\hat{\bvec{x}}^{k+1} - \bvec{z}^k)_{T^k}}_2^2 d\theta\\
\ & \le -\frac{\int_0^1 \min_{i\in T^k}M_{ii}^{-1}(u_i^k(\theta)) d\theta}{\mu}\opnorm{(\hat{\bvec{x}}^{k+1} - \bvec{z}^k)_{T^k}}_2^2.
\end{align*}
 Now, by Lemma~\ref{lem:bounds-on-iterates-zn-xhatn}, $\forall i\in T^k,\ \abs{x_i^{k+1}}\le E_{k,i}$. Consequently, $\forall i \in T^k$, $\forall \theta\in [0,1],\ M_{ii}^{-1}(u_i^k(\theta)) = \frac{1}{1-\gamma_i g_i''(\psi^{-1}(\theta\hat{x}^{k+1}_i))}\ge \frac{1}{1-\gamma\min_{u\in [\psi_i^{-1}(-E_{k,i}), \psi_i^{-1}(E_{k,i})]} g_i''(u)}=:G_{k,i}$.  Thus, \begin{align*}
\lefteqn{\int_0^1\inprod{\bvec{M}^{-1}(\bvec{u}^k(\theta))\nabla f(\bvec{x}^k)}{\hat{\bvec{x}}^{k+1} - \bvec{z}^k} d\theta} & &\\
\ & \le -\frac{G_k}{\mu}\opnorm{(\hat{\bvec{x}}^{k+1} - \bvec{z}^k)_{T^k}}_2^2, 
\end{align*}
where $G_k = \min_{i}G_{k,i}.$ Now, again using the structure of $\hat{\bvec{x}}^{k+1}$, we obtain, \begin{align*}
\lefteqn{\int_0^1\inprod{\bvec{M}^{-1}(\bvec{u}^k(\theta))\nabla f(\bvec{x}^k)}{\hat{\bvec{x}}^{k+1} - \bvec{z}^k} d\theta} & &\\
\ & \le G_k\inprod{\nabla f(\bvec{x}^k)}{\hat{\bvec{x}}^{k+1} - \bvec{z}^k}.
\end{align*}
Now the definition of $\hat{\bvec{x}}^{k+1}$ implies that \begin{align*}
\opnorm{\hat{\bvec{x}}^{k+1} - (\bvec{z}^k - \mu \nabla f(\bvec{x}^{k}))}_2^2 & \le  \opnorm{\bvec{v} - (\bvec{z}^k - \mu \nabla f(\bvec{x}^{k}))}_2^2,
\end{align*}
for any $\bvec{v}\in \real^n$, such that $\opnorm{\bvec{w}}_0\le K$. Thus, using $\bvec{z}^k$ in place of $\bvec{v}$, we obtain, \begin{align*}
\opnorm{\hat{\bvec{x}}^{k+1} - (\bvec{z}^k - \mu \nabla f(\bvec{x}^{k}))}_2^2 & \le \opnorm{\mu\nabla f(\bvec{x}^{k})}_2^2\\
\implies \inprod{\nabla f(\bvec{x}^k)}{\hat{\bvec{x}}^{k+1} - \bvec{z}^k} & \le -\frac{1}{2\mu}\opnorm{\hat{\bvec{x}}^{k+1} - \bvec{z}^k}_2^2. 
\end{align*}
Thus, we obtain, $I_2\le -\frac{G_k}{2\mu}\opnorm{\hat{\bvec{x}}^{k+1} - \bvec{z}^k}_2^2$. Therefore, $w(\hat{\bvec{x}}^{k+1}) - w(\bvec{z}^k) \le \frac{1}{2}\left(F_k^2(1+\delta_{2K}) - \frac{G_k}{\mu}\right)\opnorm{\hat{\bvec{x}}^{k+1} - \bvec{z}^k}_2^2$. Now, by the RHTP algorithm, $w(\bvec{z}^{k+1})\le w(\hat{\bvec{x}}^{k+1})\implies w(\bvec{z}^{k+1}) - w(\bvec{z}^k)\le \frac{1}{2}\left(F_k^2(1+\delta_{2K}) - \frac{G_k}{\mu}\right)\opnorm{\hat{\bvec{x}}^{k+1} - \bvec{z}^k}_2^2$. Thus, $w(\bvec{z}^{k+1})\le w(\bvec{z}^k)$, if $F_k^2(1+\delta_{2K})\le G_k/\mu $, i.e., if $\mu (1+\delta_{2K})\le (1-L_k)^2/(1-l_k)$, where $L_k = \max_i\max_{u\in [\psi_i^{-1}(-E_{k,i}),\ \psi_i^{-1}(E_{k,i})]}\gamma_i g_i''(u)$, and $l_k = \min_i\min_{u\in [\psi_i^{-1}(-E_{k,i}),\ \psi_i^{-1}(E_{k,i})]}\gamma_i g_i''(u)$. Note that since $E_{k,i}\le E_i$, we have $\psi_i^{-1}(-E_i)\le \psi_i^{-1}(-E_{k,i})$, and $\psi_i^{-1}(E_i)\ge \psi_i^{-1}(E_{k,i})$, since $\psi_i^{-1}$ is increasing for all $i=1,2,\cdots,n$. Thus, we have, $l_k\ge l$, and $L_k \le L$. Thus, if the functions $\{g_i\}_{i=1}^n$ are such that $\mu(1+\delta_{2K})<(1-L)^2/(1-l)$, then for any $k\ge 0$, we have $w(\bvec{z}^{k+1})< w(\bvec{z}^k)$. Thus the sequence $\{w(\bvec{z}^k)\}_{k\ge 0}$ is a strictly decreasing sequence of non-negative numbers (since $f$ is a non-negative valued function), and hence converges to a non-negative number. However, since $\{\bvec{z}^k\}_{k\ge 0}$ can take only finitely many values, $\{\bvec{z}^k\}_{k\ge 0}$ is eventually periodic, making $\{w(\bvec{z}^k)\}_{k\ge 0}$ eventually periodic, which means there is some finite positive integer $N_{\mathrm{it}}$, such that $w(\bvec{z}^k) = w(\bvec{z}^{N_{\mathrm{it}}}),\ \forall k\ge N_{\mathrm{it}}$. Then, we have $\hat{\bvec{x}}^{k+1} = \bvec{z}^k\implies \Lambda^{k+1} = \Lambda^k\implies \bvec{x}^{k+1} = \bvec{x}^k,\ \forall k\ge N_{\mathrm{it}}$.
\end{proof}
\subsection{Convergence rate analysis}
\label{sec:convergence-rate-analysis}
We extend the analysis of Foucart~\cite{foucart2011hard} for our case. The analysis is similar to the analysis of Foucart, with the major exception of the analysis for the diagonal matrix $\bvec{M}$ involved in our algorithm. We will require the following lemma in our analysis: 
\begin{lem}
\label{lem:inner-product-and-norm-inequality}
Let $\bvec{u},\ \bvec{v},\ \bvec{w}$ be $n$ dimensional vectors with supports $T_1,\ T_2,\ T_3$ respectively, such that $T_1,\ T_2,\ T_3\subseteq \{1,2,\cdots,\ n\}$. Let $T=T_1\cup T_2\cup T_3$. Also, let $\rho>0$ be a positive number and let $\bvec{d}(\bvec{u},\bvec{v}) = \bvec{v} - \bvec{u} - \rho\bvec{\Phi}^t\bvec{\Phi}\left(\bvec{\Psi}^{-1}(\bvec{v}) - \bvec{\Psi}^{-1}(\bvec{u})\right)$. Then, if the matrix $\bvec{\Phi}$ satisfies RIP of order $\abs{T}$, $\abs{u}_i,\ \abs{v}_i\le E$, $\forall i\in T_1\cup T_2$, and if $\rho (1+\delta_{\abs{T}})<\frac{(1-L)^2}{1-l}$, where $E$ was defined in Lemma~\ref{lem:bounds-on-iterates-zn-xhatn}, and $l,\  L$ were defined in Theorem~\ref{thm:convergence-zn}, then \begin{enumerate}
\item \begin{align}
\label{eq:inner-product-inequality}
\inprod{\bvec{w}}{\bvec{d}(\bvec{u},\bvec{v})} & \le \rho'_{\abs{T}} \opnorm{\bvec{w}}_2\opnorm{\bvec{v} - \bvec{u}}_2,
\end{align}
\item \begin{align}
\label{eq:norm-inequality}
\opnorm{(\bvec{d}(\bvec{u},\bvec{v}))_{T_3}} & \le \rho'_{\abs{T}} \opnorm{\bvec{v} - \bvec{u}}_2,
\end{align}
\end{enumerate} 
where $\rho'_{\abs{T}} = 1 - \frac{\rho (1-\delta_{\abs{T}})}{1-l}$, and $\rho'_{\abs{T}}\in (0,1)$.
\end{lem}
\begin{proof}
The proof is supplied in Appendix~\ref{sec:appendix-proof-inner-product-and-norm-inequality}.
\end{proof}
We now proceed to analyze the convergence rate of the sequence $\{\bvec{z}^k\}_{k\ge 0}$. The analysis is similar to the analysis of Theorem 3.8 of Foucart~\cite{foucart2011hard}. However, our analysis is different as the analysis takes place in a transformed domain where it requires to take care of the time varying diagonal matrix $\bvec{M}$ as defined before.
\begin{thm}
\label{thm:convergence-rate-zn}
If $\bvec{y} = \bvec{\Phi x}^\star + \bvec{e}$, under the constraint $\mu(1+\delta_{2K})<(1-L)^2/(1-l)$, the following is satisfied:
\begin{align}
\label{eq:rhtp-convergence-inequality}
\opnorm{\bvec{z}^{k+1} - \bvec{z}^\star}_2 & \le \rho_{3K}\opnorm{\bvec{z}^k - \bvec{z}^\star}_2 + \tau_{2K} \opnorm{\bvec{e}}_2.
\end{align}
\end{thm}
\begin{proof}
We start with the observation that $\nabla_{\Lambda^{k+1}}w(\bvec{z}^{k+1})=\bvec{0}_{\Lambda^{k+1}}$. Thus, $\inprod{\bvec{z} - \bvec{z}^{k+1}}{\nabla w(\bvec{z}^{k+1})} = 0$, $\forall \bvec{z}\in \real^n$, with $\support(\bvec{z}) = \Lambda^{k+1}$. Then, using $\bvec{z}^\star = \bvec{\Psi}(\bvec{x}^\star)$ and $\bvec{u} = \begin{bmatrix}
(\bvec{z}^{k+1} - \bvec{z}^\star)_{\Lambda^{k+1}}\\
\bvec{0}_{{\Lambda^{k+1}}^C}
\end{bmatrix} $, \begin{align*}
\lefteqn{\opnorm{(\bvec{z}^{k+1} - \bvec{z}^\star)_{\Lambda^{k+1}}}_2^2} & &\\
\ & = \inprod{\bvec{z}^{k+1} - \bvec{z}^\star}{\bvec{u}}\\
\ & = \inprod{\bvec{u}}{\bvec{z}^{k+1} - \bvec{z}^\star - \mu \bvec{M}(\bvec{z}^{k+1})\nabla w(\bvec{z}^{k+1})}\\
\ & = \inprod{\bvec{u}}{\bvec{z}^{k+1} - \bvec{z}^\star - \mu (\nabla f(\bvec{x}^{k+1}) - \nabla f(\bvec{x}^\star))}\\
\ & - \mu \inprod{\bvec{u}}{\nabla f(\bvec{x}^\star)}\\
\ & = \inprod{\bvec{u}}{\bvec{z}^{k+1} - \bvec{z}^\star - \mu \bvec{\Phi}^t\bvec{\Phi}(\bvec{x}^{k+1} - \bvec{x}^\star)}\\
\ & + \mu \inprod{\bvec{u}}{\bvec{\Phi}^t(\bvec{y} - \bvec{\Phi}\bvec{x}^\star)}\\
\ & = J_1 + J_2.
\end{align*}
Now, using $\bvec{x}^{k+1} = \bvec{\Psi}^{-1}(\bvec{z}^{k+1})$ and $\bvec{x}^\star = \bvec{\Psi}^{-1}(\bvec{z}^\star)$, along with Lemma~\ref{lem:inner-product-and-norm-inequality}, we get $J_1\le \mu'_{2K}\opnorm{\bvec{u}}_2 \opnorm{\bvec{z}^{k+1} - \bvec{z}^\star}_2$. On the other hand, $J_2 = \inprod{(\bvec{z}^{k+1} - \bvec{z}^\star)_{\Lambda^{k+1}}}{\bvec{\Phi}^t_{\Lambda^{k+1}}\bvec{e}}\le \sqrt{1+\delta_{K}}\opnorm{(\bvec{z}^{k+1} - \bvec{z}^\star)_{\Lambda^{k+1}}}_2\opnorm{\bvec{e}}_2$, where the last step used Cauchy-Schwarz inequality as well as the RIP. Thus, \begin{align*}
\lefteqn{\opnorm{(\bvec{z}^{k+1} - \bvec{z}^\star)_{\Lambda^{k+1}}}_2^2} & &\\
\ & \le \rho_{2K}\opnorm{(\bvec{z}^{k+1} - \bvec{z}^\star)_{\Lambda^{k+1}}}_2 \opnorm{\bvec{z}^{k+1} - \bvec{z}^\star}_2 \\
\ & + \sqrt{1+\delta_{K}}\opnorm{(\bvec{z}^{k+1} - \bvec{z}^\star)_{\Lambda^{k+1}}}_2\opnorm{\bvec{e}}_2
\end{align*}
This implies \begin{align}
\label{eq:rhtp-analysis-least-square-solution-step-inequality1}
\opnorm{(\bvec{z}^{k+1} - \bvec{z}^\star)_{\Lambda^{k+1}}}_2 & \le \mu'_{2K}\opnorm{\bvec{z}^{k+1} - \bvec{z}^\star}_2 +  \sqrt{1+\delta_{K}}\opnorm{\bvec{e}}_2
\end{align}
Then, with $A = \opnorm{\bvec{z}^{k+1} - \bvec{z}^\star}_2,\ B = \opnorm{(\bvec{z}^{k+1} - \bvec{z}^\star)_{(\Lambda^{k+1})^C}}_2$, inequality~\eqref{eq:rhtp-analysis-least-square-solution-step-inequality1} can be written as $\sqrt{A^2 - B^2}\le \rho_{2K} A +\sqrt{1+\delta_{K}}\opnorm{\bvec{e}}_2$. Consequently, solving this quadratic inequality appropriately, we obtain, \begin{align}
\opnorm{\bvec{z}^{k+1} - \bvec{z}^\star}_2 & \le \frac{1}{\sqrt{1 - (\mu'_{2K})^2}}\opnorm{(\bvec{z}^{k+1} - \bvec{z}^\star)_{(\Lambda^{k+1})^C}}_2\nonumber\\
\label{eq:rhtp-analysis-least-square-solution-step-main-inequality}
\ & + \frac{\sqrt{1+\delta_{K}}}{1-\delta_{2K}}\opnorm{\bvec{e}}_2
\end{align}
Note that for the inequality~\eqref{eq:rhtp-analysis-least-square-solution-step-main-inequality} to be valid, we require $\abs{\mu'_{2K}}<1$, which is ensured by Lemma~\ref{lem:inner-product-and-norm-inequality}. 

We now use the step $(1)$ of the RHTP algorithm. This results in \begin{align}
\lefteqn{\opnorm{(\bvec{z}^k + \mu \bvec{\Phi}^t(\bvec{y} - \bvec{\Phi}\bvec{x}^k))_{\Lambda^{k+1}}}_2} & & \nonumber\\
\ & \ge \opnorm{(\bvec{z}^k + \mu \bvec{\Phi}^t(\bvec{y} - \bvec{\Phi}\bvec{x}^k))_{\Lambda}}_2\nonumber\\
\implies \lefteqn{\opnorm{(\bvec{z}^k + \mu \bvec{\Phi}^t(\bvec{y} - \bvec{\Phi}\bvec{x}^k))_{\Lambda^{k+1}\setminus \Lambda}}_2} & & \nonumber\\
\label{eq:rhtp-analysis-hard-thresholding-step-inequality1}
\ & \ge \opnorm{(\bvec{z}^k + \mu \bvec{\Phi}^t(\bvec{y} - \bvec{\Phi}\bvec{x}^k))_{\Lambda\setminus \Lambda^{k+1}}}_2
\end{align} 
Now the LHS of inequality~\eqref{eq:rhtp-analysis-hard-thresholding-step-inequality1} can be upper bounded as below: \begin{align}
\lefteqn{\opnorm{(\bvec{z}^k + \mu \bvec{\Phi}^t(\bvec{y} - \bvec{\Phi}\bvec{x}^k))_{\Lambda^{k+1}\setminus \Lambda}}_2} & &\nonumber\\
%
\ & \le \opnorm{(\bvec{z}^k - \bvec{z}^\star - \mu \bvec{\Phi}^t\bvec{\Phi}(\bvec{x}^k - \bvec{x}^\star))_{\Lambda^{k+1}\setminus \Lambda}}_2 \nonumber\\
\label{eq:rhtp-analysis-hard-thresholding-step-lhs-upper-bound}
\ & + \mu \opnorm{(\bvec{\Phi}^t \bvec{e})_{\Lambda^{k+1}\setminus \Lambda}}_2.
\end{align}
Similarly, the RHS of inequality~\eqref{eq:rhtp-analysis-hard-thresholding-step-inequality1} can be lower bounded as below: \begin{align}
\lefteqn{\opnorm{(\bvec{z}^k + \mu \bvec{\Phi}^t(\bvec{y} - \bvec{\Phi}\bvec{x}^k))_{\Lambda\setminus \Lambda^{k+1}}}_2} & &\nonumber\\
%
\ & \ge \opnorm{(\bvec{z}^{k+1} - \bvec{z}^\star)_{\Lambda\setminus \Lambda^{k+1}}}_2\nonumber\\
\ & - \opnorm{(\bvec{z}^k - \bvec{z}^\star - \mu \bvec{\Phi}^t\bvec{\Phi}(\bvec{x}^k - \bvec{x}^\star))_{\Lambda\setminus \Lambda^{k+1}}}_2\nonumber\\
\label{eq:rhtp-analysis-hard-thresholding-step-rhs-lower-bound}
\ &  - \mu \opnorm{(\bvec{\Phi}^t\bvec{e})_{\Lambda\setminus \Lambda^{k+1}}}_2.
\end{align}
Then, writing $\bvec{a}^k = \bvec{z}^k - \bvec{z}^\star - \mu \bvec{\Phi}^t\bvec{\Phi}(\bvec{x}^k - \bvec{x}^\star)$ and using inequalities~\eqref{eq:rhtp-analysis-hard-thresholding-step-lhs-upper-bound},~\eqref{eq:rhtp-analysis-hard-thresholding-step-rhs-lower-bound}, the inequality~\eqref{eq:rhtp-analysis-hard-thresholding-step-inequality1} implies:\begin{align*}
\lefteqn{\opnorm{(\bvec{z}^{k+1} - \bvec{z}^\star)_{\Lambda\setminus \Lambda^{k+1}}}_2} & &\\
 \ & \le \opnorm{(\bvec{a}^k)_{\Lambda\setminus \Lambda^{k+1}}}_2 + \opnorm{(\bvec{a}^k)_{\Lambda\setminus \Lambda^{k+1}}} \\
\ & + \mu \opnorm{(\bvec{\Phi}^t\bvec{e})_{\Lambda\setminus \Lambda^{k+1}}}_2 + \mu \opnorm{(\bvec{\Phi}^t\bvec{e})_{\Lambda^{k+1}\setminus\Lambda}}_2.
\end{align*}
Then, using $A+B\le \sqrt{2(A^2+B^2)}$ for non-negative $A,B$, we obtain:\begin{align*}
\lefteqn{\opnorm{(\bvec{z}^{k+1} - \bvec{z}^\star)_{\Lambda\setminus \Lambda^{k+1}}}_2} & &\\
\ & \le \sqrt{2}\opnorm{(\bvec{a}^k)_{\Lambda \Delta \Lambda^{k+1}}}_2 + \sqrt{2}\mu \opnorm{(\bvec{\Phi}^t\bvec{e})_{\Lambda \Delta \Lambda^{k+1}}}_2.
\end{align*}
Now, using Lemma~\ref{lem:inner-product-and-norm-inequality}, we find that $\opnorm{(\bvec{a}^k)_{\Lambda \Delta \Lambda^{k+1}}}_2\le \mu'_{3K}\opnorm{\bvec{z}^k - \bvec{z}^\star}_2$. Consequently,
\begin{align}
\lefteqn{\opnorm{(\bvec{z}^{k+1} - \bvec{z}^\star)_{\Lambda\setminus \Lambda^{k+1}}}_2} & &\nonumber\\
\ & \le \sqrt{2}\mu'_{3K}\opnorm{\bvec{z}^k - \bvec{z}^\star}_2 + \sqrt{2}\mu \opnorm{(\bvec{\Phi}^t\bvec{e})}_2\nonumber\\
\label{eq:rhtp-analysis-hard-thresholding-step-main-inequality}
\ & \le \sqrt{2}\mu'_{3K}\opnorm{\bvec{z}^k - \bvec{z}^\star}_2 + \sqrt{2}\mu \sqrt{1+\delta_{2K}}\opnorm{\bvec{e}}_2.
\end{align}
Finally, combining inequalities~\eqref{eq:rhtp-analysis-least-square-solution-step-main-inequality} and \eqref{eq:rhtp-analysis-least-square-solution-step-main-inequality}, we arrive at the desired inequality following inequality. Note that $\rho_{3K}<1$ if $\mu'_{3K}<1/\sqrt{3}$, since $\mu'_{3K}>\mu'_{2K}$. Thus, $\rho_{3K}<1$ if $\mu > \frac{\left(1 - \frac{1}{\sqrt{3}}\right)(1-l)}{1-\delta_{3K}}$. Thus, if $l > 0$, the lower bound on $\mu$ which suffices for the algorithm to have a linear convergence rate, is smaller than the one obtained for HTP, which corresponds to the case $l = 0$.
\end{proof}
\subsection{Analysis of the number of iteration required by RHTP for correct support recovery}
\label{sec:analysis-number-of-iteration}
As suggested by Theorem~\ref{thm:convergence-zn}, the RHTP algorithm converges in finite number of iterations. In this section, we proceed to estimate the number of iterations that RHTP takes to recover the support of a $K$-sparse vector $\bvec{x}^\star$ from a measurement vector $\bvec{y} = \bvec{\Phi x} + \bvec{e}$. We present two results that estimate this iteration complexity, one using the knowledge of some structure of the unknown signal $\bvec{x}^\star$, and another which is independent of the structure of $\bvec{x}^\star$. The analyses presented are similar to the analyses of Corollary 3.6 of~\cite{foucart2011hard} and Lemma 3 and Theorem 5 of~\cite{bouchot2016hard}. However, the main novelty in our analyses lie in the way we take care of the effect of the time-varying diagonal matrix $\bvec{M}$ in our analyses.
\begin{thm}
\label{thm:no-iteration-signal-dependent-estimate}
Suppose that $\rho_{3K}<1$. Let $\tau'_1 = \sqrt{2(1+\delta_2)}\mu + \frac{\tau_{2K}}{1-\rho_{3K}}$, and let the error $\bvec{e}$ is such that $\tau_1\opnorm{\bvec{e}}_2 \le \abs{z^\star_{\min}}$. Then, the RHTP algorithm recovers the correct support and converges in at most $\left\lceil \frac{\left(\sqrt{2}\mu'_{3K}\opnorm{\bvec{z}^0 - \bvec{z}^\star}_2\right)/\left(\abs{z_{\min}^\star} - \tau_1\opnorm{\bvec{e}}_2\right) }{\ln\left(1/\rho_{3K}\right)} \right\rceil$ iterations.
\end{thm}
\begin{proof}
The RHTP algorithm stops in $k$ iterations recovering the true support $\Lambda$, if $\Lambda^{k-1} = \Lambda$, and $\Lambda^{k} = \Lambda$. Now, we first find conditions sufficient to ensure $\Lambda^k = \Lambda$. This is ensured if $\forall i\in \Lambda$, and $\forall j\in \Lambda^C$, \begin{align}
\label{eq:finite-no-convergence-signal-dependent-bound-preliminary-inequality}
\abs{(\bvec{z}^k + \mu \bvec{\Phi}^t (\bvec{y} - \bvec{\Phi x}^k))_i} > \abs{(\bvec{z}^k + \mu \bvec{\Phi}^t (\bvec{y} - \bvec{\Phi x}^k))_j}.
\end{align}
The LHS of inequality~\eqref{eq:finite-no-convergence-signal-dependent-bound-preliminary-inequality} can be lower bounded as below: \begin{align}
\lefteqn{\abs{(\bvec{z}^k + \mu \bvec{\Phi}^t (\bvec{y} - \bvec{\Phi x}^k))_i}}  & & \nonumber\\
\label{eq:finite-no-convergence-signal-dependent-bound-lhs-lower-bound}
\ & \ge \abs{z_{\min}^\star} - \abs{(\bvec{z}^k - \bvec{z}^\star - \mu \bvec{\Phi}^t \bvec{\Phi}(\bvec{x}^\star - \bvec{x}^k))_i} - \mu\abs{(\bvec{\Phi}^t\bvec{e})_i}.
\end{align}
On the other hand, the RHS of inequality~\eqref{eq:finite-no-convergence-signal-dependent-bound-preliminary-inequality} can be upper bounded as below:\begin{align}
\lefteqn{\abs{(\bvec{z}^k + \mu \bvec{\Phi}^t (\bvec{y} - \bvec{\Phi x}^k))_j}}  & & \nonumber\\
\label{eq:finite-no-convergence-signal-dependent-bound-rhs-upper-bound}
\ & \le \abs{(\bvec{z}^k - \bvec{z}^\star - \mu \bvec{\Phi}^t \bvec{\Phi}(\bvec{x}^\star - \bvec{x}^k))_j} + \mu\abs{(\bvec{\Phi}^t\bvec{e})_j}.
\end{align}
Thus using inequalities~\eqref{eq:finite-no-convergence-signal-dependent-bound-lhs-lower-bound}, and~\eqref{eq:finite-no-convergence-signal-dependent-bound-rhs-upper-bound}, one can see that the inequality~\eqref{eq:finite-no-convergence-signal-dependent-bound-preliminary-inequality} is ensured if \begin{align*}
\abs{z_{\min}^\star} & \ge \abs{(\bvec{z}^k - \bvec{z}^\star - \mu \bvec{\Phi}^t \bvec{\Phi}(\bvec{x}^\star - \bvec{x}^k))_i}\\
\ & + \abs{(\bvec{z}^k - \bvec{z}^\star - \mu \bvec{\Phi}^t \bvec{\Phi}(\bvec{x}^\star - \bvec{x}^k))_j}\\
\ & + \mu\abs{(\bvec{\Phi}^t\bvec{e})_i} + \mu\abs{(\bvec{\Phi}^t\bvec{e})_j},
\end{align*}
or equivalently, if \begin{align*}
\abs{z_{\min}^\star} & \ge \sqrt{2}\abs{(\bvec{z}^k - \bvec{z}^\star - \mu \bvec{\Phi}^t \bvec{\Phi}(\bvec{x}^\star - \bvec{x}^k))_{\{i,j\}}}\\
\ & + \sqrt{2}\mu\abs{(\bvec{\Phi}^t\bvec{e})_{\{i,j\}}},
\end{align*}
Now, using Lemma~\ref{lem:inner-product-and-norm-inequality} and RIP,\begin{align*}
\ & \abs{(\bvec{z}^k - \bvec{z}^\star - \mu \bvec{\Phi}^t \bvec{\Phi}(\bvec{x}^\star - \bvec{x}^k))_{\{i,j\}}}\\
\ & + \mu\abs{(\bvec{\Phi}^t\bvec{e})_{\{i,j\}}}\\
\ & \le \mu'_{2K+1}\opnorm{\bvec{z}^k - \bvec{z}^\star}_2 + \sqrt{(1+\delta_2)}\mu \opnorm{\bvec{e}}_2\\
\ & \le \mu'_{3K}\opnorm{\bvec{z}^k - \bvec{z}^\star}_2 + \sqrt{(1+\delta_2)}\mu \opnorm{\bvec{e}}_2\\
\ & \le \sqrt{1-(\mu'_{2K})^2}\rho_{3K} \opnorm{\bvec{z}^k - \bvec{z}^\star}_2 + \sqrt{(1+\delta_2)}\mu \opnorm{\bvec{e}}_2
\end{align*} or equivalently, if \begin{align*}
\abs{z_{\min}^\star} & \ge \sqrt{2}\mu'_{3K}\opnorm{\bvec{z}^k - \bvec{z}^\star}_2 + \sqrt{2(1+\delta_2)}\mu \opnorm{\bvec{e}}_2.
\end{align*}
Since $\opnorm{\bvec{z}^k - \bvec{z}^\star}_2\le \rho_{3K}^n\opnorm{\bvec{z}^0 - \bvec{z}^\star}_2 + \frac{\tau}{1-\rho_{3K}}\opnorm{\bvec{e}}_2$, the previous inequality is satisfied if \begin{align*}
\abs{z_{\min}^\star} & \ge \sqrt{2}\mu'_{3K}\rho_{3K}^k \opnorm{\bvec{z}^0 - \bvec{z}^\star}_2 + \tau_1 \opnorm{\bvec{e}}_2.
\end{align*}
 Thus, if $\tau_1\opnorm{\bvec{e}}_2<\abs{z^\star_{\min}},$ then RHTP recovers the correct support in $\left\lceil \frac{\left(\sqrt{2}\mu'_{3K}\opnorm{\bvec{z}^0 - \bvec{z}^\star}_2\right)/\left(\abs{z_{\min}^\star} - \tau_1\opnorm{\bvec{e}}_2\right) }{\ln\left(1/\rho_{3K}\right)} \right\rceil$ iterations.
\end{proof}


We now analyze the number of iteration required for perfect support recovery of the unknown signal, without requiring any knowledge about the minimum absolute value of the unknown signal. The analysis follows from the Lemmas~$3$ and Theorems~$5$ and $6$ of Bouchot~\emph{et al}~\cite{bouchot2016hard}. However, the presence of the diagonal matrix $\bvec{M}$ makes our analysis different.

\begin{lem}
\label{lem:condition-on-x-finite-convergence}
Suppose that the RHTP algorithm has produced a sequence of index sets $\{\Lambda^k\}_{k\ge 0}$, with measurement vector $\bvec{y} = \bvec{\Phi}\bvec{x}^\star + \bvec{e}$. For integers $n,p\ge 0$, let $\Lambda^k$ contains the indices of the $p$ largest absolute entries of $\bvec{z}^\star$. Then, for $k',q\ge 1$, $\Lambda^{k+k'}$ will contain indices of the $p+q$ largest absolute entries of $\bvec{z}^\star$ if the following is satisfied:
\begin{align}
\label{eq:no-of-iteration-lemma-inequality}
r_{p+q}(\bvec{z}^\star)>\rho_{3K}^{k'} \opnorm{r(\bvec{z}^\star)_{\{p+1,\cdots,K\}}}_2 + \kappa_{3K}\opnorm{\bvec{e}}_2
\end{align}  
where \begin{align*}
\kappa_{3K} = \frac{\sqrt{2(1+\delta_{3K})}(\mu'_{3K} + 1 - \delta_{3K})}{1-\delta_{3K}} + \frac{\sqrt{2}\mu'_{3K}\tau_{3K}}{1-\rho_{3K}}
\end{align*}
\end{lem}
\begin{proof}
Let $\pi$ be the permutation of $\{1,2,\cdots,\ n\}$ such that $\abs{z^\star_{\pi(j)}} = r_j(\bvec{z}^\star)$ for all $j=1,2,\cdots,\ n$. We now require to find conditions that ensure $\pi(\{1,\cdots,p+q\})\subseteq \Lambda^{k+k'}$, given that $\pi(\{1,\cdots,p\})\subseteq \Lambda^{k}$. To ensure that $\pi(\{1,\cdots,p+q\})\subseteq \Lambda^{k+k'}$, it is enough to ensure  
 \begin{align}
\lefteqn{\min_{j\in \{1,\cdots,p+q\}}\abs{(\bvec{z}^{k+k'-1} + \mu \bvec{\Phi}^t(\bvec{y}-\bvec{\Phi x}^{k+k'-1}))_{\pi(j)}}} & &\nonumber\\
\label{eq:no-iteration-sufficient-temp-condition1}
\ & > \max_{i\in \Lambda^C}\abs{(\bvec{z}^{k+k'-1} + \mu \bvec{\Phi}^t(\bvec{y}-\bvec{\Phi x}^{k+k'-1}))_{i}}.
\end{align}
Further proceeding along the lines of Bouchot \emph{et al}~\cite{bouchot2016hard} and using steps similar to those used to arrive at the inequality (13) therein, we find that the inequality~\eqref{eq:no-iteration-sufficient-temp-condition1} is ensured if, $\forall j\in \{1,2,\cdots,p+q\}$, and $i\in \Lambda^C$, the following is satisfied \begin{align}
r_{p+q}(\bvec{z}^\star) & >\abs{(\bvec{z}^{k+k'-1} - \bvec{z}^\star - \mu\bvec{\Phi}^t\bvec{\Phi}(\bvec{x}^{k+k'-1} - \bvec{x}^\star) + \mu\bvec{\Phi}^t\bvec{e})_{\pi(j)}}\nonumber\\
\label{eq:no-iteration-sufficient-temp-condition2}
\ & + \abs{(\bvec{z}^{k+k'-1} - \bvec{z}^\star - \mu\bvec{\Phi}^t\bvec{\Phi}(\bvec{x}^{k+k'-1} - \bvec{x}^\star) + \mu\bvec{\Phi}^t\bvec{e})_{i}}
\end{align} 
Now, the RHS of the inequality~\eqref{eq:no-iteration-sufficient-temp-condition2} can be upper bounded as below: \begin{align}
\lefteqn{\sqrt{2}\opnorm{(\bvec{z}^{k+k'-1} - \bvec{z}^\star - \mu\bvec{\Phi}^t\bvec{\Phi}(\bvec{x}^{k+k'-1} - \bvec{x}^\star) + \mu\bvec{\Phi}^t\bvec{e})_{\pi(j),\{i\}}}_2} & & \nonumber\\
\ & \le \sqrt{2}\opnorm{(\bvec{z}^{k+k'-1} - \bvec{z}^\star - \mu\bvec{\Phi}^t\bvec{\Phi}(\bvec{x}^{k+k'-1} - \bvec{x}^\star))_{\{\pi(j),i\}}}_2 \nonumber\\
\ & + \sqrt{2}\mu\opnorm{(\bvec{\Phi}^t\bvec{e})_{\{\pi(j),i\}}}_2\nonumber\\
\label{eq:no-iteration-sufficient-temp-condition3}
\ & \stackrel{(f) }{\le}\sqrt{2}\mu'_{2K+2}\opnorm{\bvec{z}^{k+k'-1}-\bvec{z}^\star}_2 + \sqrt{2}\sqrt{1+\delta_2}\opnorm{\bvec{e}}_2\nonumber\\
\ & \stackrel{(g)}{\le}\sqrt{2}\mu'_{3K}\left(\rho_{3K}^{k'-1}\opnorm{\bvec{z}^k - \bvec{z}^\star}_2 + \frac{\tau_{2K}\opnorm{\bvec{e}}_2}{1-\rho_{3K}}\right) + \sqrt{2(1+\delta_{2K})}\opnorm{\bvec{e}}_2
\end{align}
where step $(f)$ uses inequality~\eqref{eq:norm-inequality} of Lemma~\ref{lem:inner-product-and-norm-inequality} and step $(g)$ uses inequality~\eqref{eq:rhtp-convergence-inequality} of Theorem~\ref{thm:convergence-rate-zn} $k'-1$ times. Now using~\eqref{eq:rhtp-analysis-least-square-solution-step-main-inequality} and the assumption that $\pi(\{1,\cdots,p\})\subseteq \Lambda^k $, we find \begin{align*}
\lefteqn{\sqrt{2}\opnorm{(\bvec{z}^{k+k'-1} - \bvec{z}^\star - \mu\bvec{\Phi}^t\bvec{\Phi}(\bvec{x}^{k+k'-1} - \bvec{x}^\star) + \mu\bvec{\Phi}^t\bvec{e})_{\pi(j),\{i\}}}_2} & &\\
\ & \le \sqrt{2}\mu'_{3K}\rho_{3K}^{k'-1}\left(\frac{\opnorm{\bvec{z}^\star}_{(\Lambda^k)^C}}{\sqrt{1-(\mu'_{2K})^2}}+\frac{\sqrt{1+\delta_{K}}}{1-\delta_{2K}}\opnorm{\bvec{e}}_2\right) \\
\ & + \left(\frac{\sqrt{2}\mu'_{3K}\tau_{2K}}{1-\rho_{3K}} + \sqrt{2(1+\delta_{2K})}\right)\opnorm{\bvec{e}}_2\\
\ & \le \rho_{3K}^{k'}\opnorm{\bvec{z}^\star_{\pi(\{1,\cdots,p\})^C}}_2\\
\ & + \left(\frac{\sqrt{2(1+\delta_K)}\mu'_{3K}\rho_{3K}^{k'-1}}{1-\delta_{2K}} + \frac{\sqrt{2}\mu'_{3K}\tau_{2K}}{1-\rho_{3K}} + \sqrt{2(1+\delta_{2K})}\right)\opnorm{\bvec{e}}_2\\
\ & \le \rho_{3K}^{k'} \opnorm{r(\bvec{z}^\star)_{\{p+1,\cdots,K\}}}_2 + \kappa_{3K}\opnorm{\bvec{e}}_2.
\end{align*}
Thus,~\eqref{eq:no-iteration-sufficient-temp-condition2} is satisfied as soon as the condition in~\eqref{eq:no-of-iteration-lemma-inequality} is satisfied, which concludes the proof.
\end{proof}
We now present the main theorem on the number of iterations for convergence of HTP, which, unlike Theorem~\ref{thm:no-iteration-signal-dependent-estimate}, does not require the knowledge of the unknown signal. 
\begin{thm}
Under the condition $\frac{(1-\frac{1}{\sqrt{3}})(1-l)}{1-\delta_{3K}}<\mu<\frac{(1-L)^2}{(1-l)(1+\delta_{2K})}$, every sparse vector $\bvec{x}^\star\in \real^n$ is recovered using the measurement vector $\bvec{y}=\bvec{\Phi}\bvec{x}^\star$ in no more than $cK$ iterations, where $c=\frac{\ln(4/\rho_{3K}^2)}{\ln(1/\rho_{3K}^2)}$. 
\end{thm}
\begin{proof}
The proof of this theorem is exactly the same as the proof of Theorem 5 of~\cite{bouchot2016hard} with just replacing $\bvec{x}$ and $\rho$ in the theorem by $\bvec{z}^\star$ and $\rho_{3K}$ respectively. 
\end{proof}
\section{Numerical experiments}
\vspace{-1mm}
\label{sec:numerical-experiments}
In the simulation experiments, the unknown vector has dimensions $n=512,\ m=256$. The entries of the unknown signal corresponding to its support are generated independently according to $\mathcal{N}(0,1)$. The measurements are assumed to be noiseless i.e. we assume $\bvec{y}=\bvec{\Phi x}^\star$. The entries of the measurement matrix $\bvec{\Phi}$ are generated according to i.i.d.~$\mathcal{N}(0,m^{-1})$, and the entries of $\bvec{x}^\star$ are generated according to i.i.d.~$\mathcal{N}(0,1)$. The regularizer  is taken as $\sum_{j=1}^N \gamma_jg(x_j)$, where $g(x) =(x^2+\epsilon^2)^{q/2},\ \gamma_j=\gamma,\ \forall j$. For $q>1$, this regularizer puts a constraint on the energy of entries in the support of the estimated vector. For $q\le 1$, the regularizer promotes sparsity among the entries of the support of the estimated vector, so that, intuitively, estimation of the entries with large magnitude of the support of $\bvec{x}^\star$ will get more preference than the entries with very small magnitude. We fix $\epsilon = 1.4\gamma,\ \gamma=0.3$. The values of $q$ are chosen from the set $\{0.5,1,1.5,2\}$, and we take $\mu=0.3$. For the simulation, $100$ ensembles of randomly generated measurement matrices and random unknown signals are generated. The HTP and RHTP algorithms are run for $100$ iterations using each of these ensembles. We say that an algorithm successfully recovers the true sparse signal if the norm of the error between the estimate produced at an iteration of the algorithm and the actual signal is less than $10^{-6}$. If this happens at an iteration $N_{\mathrm{it}}$, the algorithm stops at that iteration and we say that the algorithm has taken $N_{\mathrm{it}}$ iterations to perfectly recover the unknown sparse vector.
\begin{figure}[t!]
	\centering
	\includegraphics[height=1.5in,width=3in]{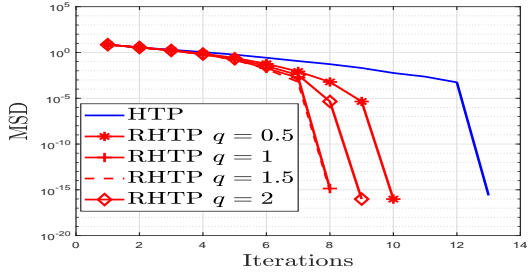}
	\caption{MSD vs iterations for RHTP and HTP, $m=256,\ K=51$}
	\label{fig:msd-vs-iterations-rhtp-htp}
\end{figure}
\begin{figure}[t!]
	\centering
	\begin{subfigure}{0.5\textwidth}
		\centering
		\includegraphics[height=1.5in, width=3in]{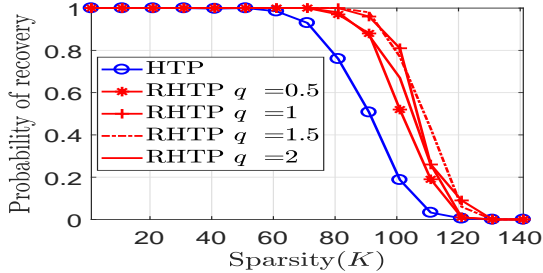}
		\caption{}
		\label{fig:prob-vs-sparsity}
	\end{subfigure}
	\begin{subfigure}{0.5\textwidth}
		\centering
		\includegraphics[height=1.5in, width=3in]{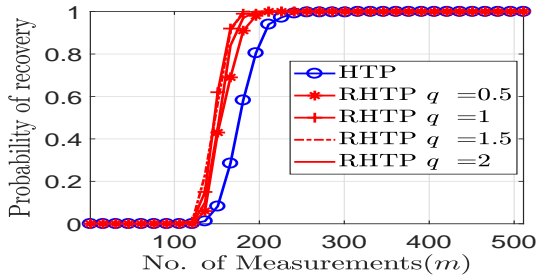}
		\caption{}
		\label{fig:prob-vs-no-measurements}
	\end{subfigure}
	\caption{Probability of recovery vs sparsity in Fig.~\ref{fig:prob-vs-sparsity} ($m=256$) and vs no. of measurements in Fig.~\ref{fig:prob-vs-no-measurements} ($K=51$).}
	\label{fig:prob-recovery-rhtp-vs-htp}
\end{figure}
\begin{figure}[t!]
	\centering
	\begin{subfigure}{0.5\textwidth}
		\centering
		\includegraphics[height=1.5in, width=3in]{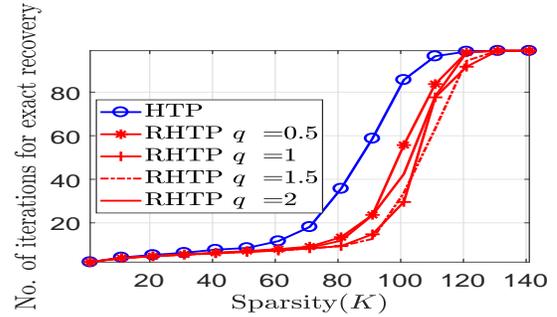}
		\caption{}
		\label{fig:it-vs-sparsity}
	\end{subfigure}
	\begin{subfigure}{0.5\textwidth}
		\centering
		\includegraphics[height=1.5in, width=3in]{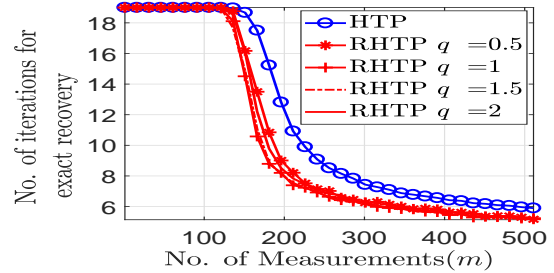}
		\caption{}
		\label{fig:it-vs-no-measurements}
	\end{subfigure}
	\caption{Average number of iteration for perfect recovery vs sparsity in Fig.~\ref{fig:it-vs-sparsity} ($m=256$) and vs no. of measurements in Fig.~\ref{fig:it-vs-no-measurements} ($K=51$).}
	\label{fig:no-iteration-no-rhtp-vs-htp}
\end{figure}

In Fig.~\ref{fig:msd-vs-iterations-rhtp-htp} the decay of the iterates of RHTP, for different $q$, as well as HTP are plotted against iteration number. Clearly, RHTP has faster convergence than HTP, with the fastest convergence seen for $q=1,1.5$.
In Fig.~\ref{fig:prob-recovery-rhtp-vs-htp}, the probability of recovery of HTP and RHTP for various values of $q$ are plotted against both sparsity $K$ (in Fig.~\ref{fig:prob-vs-sparsity}) and number of measurements $m$ (in Fig.~\ref{fig:prob-vs-no-measurements}). Clearly, for all $q$, RHTP has better probability of recovery performance. In Fig.~\ref{fig:no-iteration-no-rhtp-vs-htp}, the number of iterations required for convergence for both RHTP and HTP are plotted against $K$ (Fig.~\ref{fig:it-vs-sparsity}) and $m$ (Fig.~\ref{fig:it-vs-no-measurements}). Again, in all these plots the number of iterations required by RHTP to converge can be seen to smaller than that of HTP, with the best performance shown by $q=1,1.5$.
\appendices
\section{Proof of Lemmas}
\label{sec:appendix-proof-lemmas}
%
\subsection{Proof of Lemma~\ref{lem:norm-of-projection-upper-bound} }
\label{sec:appendix-proof-of-lemma-norm-of-projection-upper-bound}
First note that $\opnorm{\proj{S}\bvec{y}}_2^2 = \inprod{\proj{S}\bvec{y}}{\proj{S}\bvec{y}} = \inprod{\proj{S}\bvec{y}}{\bvec{y}}=\inprod{\bvec{\Phi}_{S}\bvec{u}_{S}}{\bvec{\Phi}_\Lambda\bvec{x}_\Lambda}=\bvec{u}_{T}^t\bvec{\Phi}^t_T\bvec{\Phi}_T\bvec{x}_T$, where $T = S\cup \Lambda$, and $\bvec{u}$ is a vector with support $S$, and $\bvec{u}_S$ is the projection coefficients after projecting $\bvec{y}$ to the span of $\bvec{\Phi}_{S}$. Now, as the matrix $\bvec{\Phi}$ satisfies RIP of order $\abs{S\cup \Lambda}$, the matrix $\bvec{\Phi}^t_T\bvec{\Phi}_T$ is positive definite with the maximum and minimum eigenvalues (let $\lambda_M$ and $\lambda_m$ respectively) bounded as $0<1-\delta_{\abs{T}}<\lambda_m\le \lambda_M < 1 + \delta_{\abs{T}}$. Now, $\inprod{\bvec{u}}{\bvec{x}}=0$ since $S\cap \Lambda=\emptyset$. Consequently, applying Wielandt's Lemma~\ref{lem:wielandt-theorem}, we arrive at \begin{align*}
\opnorm{\proj{S}\bvec{y}}_2^2 & \le \delta_{\abs{T}}\opnorm{\bvec{\Phi}_S\bvec{u}_S}_2\opnorm{\bvec{\Phi}_\Lambda\bvec{x}_\Lambda}_2\\
\ & = \delta_{\abs{T}}\opnorm{\proj{S}\bvec{y}}_2\opnorm{\bvec{\Phi}_\Lambda\bvec{x}_\Lambda}_2,
\end{align*} 
which subsequently produces the desired result after canceling $\opnorm{\proj{S}\bvec{y}}$ once from both sides of the above inequality.
\subsection{Proof of Lemma~\ref{lem:inner-product-of-projections-upper-bound} }
\label{sec:appendix-proof-of-lemma-inner-product-of-projections-upper-bound}
We begin the proof by writing $\bvec{\phi}_i^t\dualproj{S}\bvec{y}$ as $\bvec{\phi}_i^t\dualproj{S}\bvec{\Phi}_{\Lambda\setminus S}\bvec{x}_{\Lambda\setminus S}=\inprod{\dualproj{S}\bvec{\phi}_i}{\bvec{\Phi}_{\Lambda\setminus S}\bvec{x}_{\Lambda\setminus S}}=\inprod{\bvec{\Phi}_{S\cup \{i\}}\bvec{u}_{S\cup \{i\}}}{\bvec{\Phi}_{\Lambda\setminus S}\bvec{x}_{\Lambda\setminus S}} = \inprod{\bvec{\Phi}_{S\cup \{i\}\cup \Lambda} \bvec{u}}{\bvec{\Phi}_{S\cup \{i\}\cup \Lambda} \bvec{v}}$, where $\bvec{u},\ \bvec{v}\in \real^{\abs{S\cup \{i\}\cup \Lambda}}$, such that \begin{align*}
u_j & = \left\{\begin{array}{ll}
0, & \mbox{if}\ j\in (S\cup \{i\})^C\\
1, & \mbox{if}\ j=i\\
-w_j, & \mbox{if}\ j\in S,
\end{array}\right.
\end{align*} and  \begin{align*}
v_j & =\left\{\begin{array}{ll}
0, & \mbox{if}\ j\in (\Lambda\setminus S)^C\\
x_j, & \mbox{if}\ j \in \Lambda\setminus S.
\end{array}
\right.
\end{align*} Here $\bvec{w}_S$ is the vector such that $\bvec{\Phi}_S\bvec{w}_S = \proj{S}\bvec{\phi}_i$. Such a vector $\bvec{w}_S$ is unique since the matrix $\bvec{\Phi}_S$ satisfies RIP of order $\abs{S}$ which is a consequence of the assumption that the matrix $\bvec{\Phi}$ satisfies RIP of order $\abs{S} + \abs{\Lambda} +1$. Now, we can apply Wielandt's Lemma~\ref{lem:wielandt-theorem} along with the RIP to obtain $\abs{\bvec{\phi}_i^t\dualproj{S}\bvec{y}}\le \delta_{\abs{S}+\abs{\Lambda}+1}\opnorm{\bvec{\Phi}_{S\cup \{i\}\cup \Lambda} \bvec{u}}_2\opnorm{\bvec{\Phi}_{S\cup \{i\}\cup \Lambda} \bvec{v}}_2 = \delta_{\abs{S}+\abs{\Lambda}+1}\opnorm{\dualproj{S}\bvec{\phi}_i}_2\opnorm{\bvec{\Phi}_{\Lambda\setminus S}\bvec{x}_{\Lambda\setminus S}}_2 = \delta_{\abs{S}+\abs{\Lambda}+1}\opnorm{\dualproj{S}\bvec{\phi}_i}_2\opnorm{\bvec{y}}_2$ since $\Lambda\cap S=\emptyset$.
\subsection{Proof of Lemma~\ref{lem:homeomorphism-Psi} }
\label{sec:appendix-proof-lemma-homemorphism-Psi}
Since for each $i=1,2,\cdots,\ n$, the $i^{\mathrm{th}}$ coordinate of $\bvec{\Psi}(\bvec{x})$ is only a function of $x_i$ (namely $\psi(x_i)$), to prove that $\bvec{\Psi}(\bvec{x})$ is a homeomorphism it is enough to prove that, for all $j=1,2,\cdots,n$, $\psi_j:\real\to \real$ is bijective, and that both $\psi_j$ and $\psi_j^{-1}$ are continuous. The continuity of $\psi_j$ follows from the observation that $\psi_j(x) = x - \gamma_j g_j'(x)$ and that $g_j'$ is continuous since $g_j\in \mathcal{C}^2$. Also, as $\gamma_j g_j''(x) < 1\ \forall x\in \real$, $\psi_j$ is strictly increasing, which makes it bijective with a continuous inverse.
\subsection{Proof of Lemma~\ref{lem:jacobians-of-psi-and-psi-inverse} }
\label{sec:appendix-proof-lemm-jacobians-of-psi-and-psi-inverse}
The proof follows by observing that $[\bvec{D\Psi}(\bvec{x})]_{ij} = \frac{\partial \psi_i(x_i)}{\partial x_j} = \psi_i'(x_i)\delta_{ij}$ and similarly, $[\bvec{D\Psi}^{-1}(\bvec{x})]_{ij} = \frac{\partial \psi_i^{-1}(x_i)}{\partial x_j} = \frac{1}{\psi_i'(\psi_i^{-1}(x_i))}\delta_{ij}$, and subsequently recalling the definitions of $\bvec{M}(\bvec{x})$ and $\bvec{M}^{-1}(\bvec{x})$ in Section~\ref{sec:theoretical-results}.
\subsection{Proof of Lemma~\ref{lem:bounds-on-iterates-zn-xhatn} }
\label{sec:appendix-proof-lemma-bounds-on-iterates-zn-xhatn}
First note that $\forall k\ge 1$, $\bvec{x}^k = \bvec{\Phi}_{\Lambda^k}^\dagger \bvec{y}$. Thus, $\proj{\Lambda^k}\bvec{y} = \bvec{\Phi}_{\Lambda^k}\bvec{x}^k$. Thus, using RIP and the fact that $\opnorm{\proj{\Lambda^k}\bvec{y}}_2\le \opnorm{\bvec{y}}_2$, we have  $\opnorm{\bvec{x}^k}_2\le \frac{\opnorm{\bvec{y}}_2}{\sqrt{1-\delta_K}}$. However, when $\Lambda^{k}\cap \Lambda= \emptyset$, for the noiseless case, $\bvec{y} = \bvec{\Phi}_\Lambda\bvec{x}_\Lambda$. Thus, using Lemma~\ref{lem:norm-of-projection-upper-bound} and using RIP one obtains, $\opnorm{\bvec{x}^k}_2\le \frac{\delta_{2K}\opnorm{\bvec{y}}_2}{\sqrt{1-\delta_K}}$. Thus, for noiseless case we write $\opnorm{\bvec{x}^k}_2\le B_n$.

 Now, to find an upper bound on $\opnorm{\bvec{z}^k}_2$, we observe that $\bvec{z}^k = \bvec{\Psi}(\bvec{x}^k) = \bvec{D\Psi}(\theta\bvec{x}^k)\bvec{x}^k$ for some $\theta\in [0,1]$, where we have used the mean value theorem and the fact that $\bvec{\Psi}(\bvec{0}) = \bvec{0}$. Using Lemma~\ref{lem:jacobians-of-psi-and-psi-inverse} we then deduce that $\bvec{z}^k = \bvec{M}(\bvec{\Psi}(\theta\bvec{x}^k))\bvec{x}^k\implies z_i^k = [\bvec{M}(\bvec{\Psi}(\theta\bvec{x}^k))]_{ii} x_i^k$. Hence, $\abs{z_i^k} = [\bvec{M}(\bvec{\Psi}(\theta\bvec{x}^k))]_{ii}\abs{x_i^k} \le  (1-\gamma_i g_i''(\theta x_i^k))B_k\le C_{k,i}B_k$, since $1-\gamma_i g_i''(\theta x_i^k)\le 1-\gamma_i \min_{u\in [-B_k, B_k]}g_i''(u)=: C_{k,i}$ as $\forall i=1,2,\cdots,\ n$, $\abs{\theta x_i^k}\le \opnorm{\bvec{x}^k}_2\le B_k$. 

To find a bound on $\hat{x}_i^{k+1}$, we first observe that \begin{align*}
\hat{x}_i^{k+1} & =\left\{
\begin{array}{ll}
{z}_i^k, & \mbox{if}\ i\in \Lambda^{k+1}\cup \Lambda^k\\
-\mu\partial_i f(\bvec{x}^k), & \mbox{if}\ i\in \Lambda^{k+1}\setminus \Lambda^k\\
0, & \mbox{else},
\end{array}
\right.
\end{align*}
which is a consequence of the fact that $\nabla_{\Lambda^k}f(\bvec{x}^k) = \bvec{0}_{\Lambda^k}$.
Thus, if $i\in \Lambda^{k+1}\cap \Lambda^k$, $\abs{\hat{x}^{k+1}_i}\le B_kC_{k,i}$. On the other hand, if $i\in \Lambda^{k+1}\setminus \Lambda^k$, $\abs{\hat{x}^{k+1}_i}=\mu\abs{\partial_i f(\bvec{x}^k)} = \mu\abs{\bvec{\phi}_i^t(\bvec{y} - \bvec{\Phi x}^k)}$. Since, $\bvec{\Phi x}^k = \proj{\Lambda^k}\bvec{y}$, we have $\bvec{\phi}_i^t(\bvec{y} - \bvec{\Phi x}^k) = \bvec{\phi}_i^t\dualproj{\Lambda^k}\bvec{y}$. Using Cauchy-Schwartz inequality it then trivially follows that $\abs{\bvec{\phi}_i^t(\bvec{y} - \bvec{\Phi x}^k)}\le \max_i\opnorm{\bvec{\phi}_i}_2\opnorm{\bvec{y}}_2.$ However, when $\Lambda^{k+1}\cap \Lambda=\emptyset$, for the noiseless case, Lemma~\ref{lem:inner-product-of-projections-upper-bound} allows to deduce the following upper bound: $\abs{\bvec{\phi}_i^t(\bvec{y} - \bvec{\Phi x}^k)}\le \delta_{2K+1}\max_i\opnorm{\bvec{\phi}_i}_2\opnorm{\bvec{y}}_2$.
\subsection{Proof of Lemma~\ref{lem:inner-product-and-norm-inequality} }
\label{sec:appendix-proof-inner-product-and-norm-inequality}
First note that $\inprod{\bvec{w}}{\bvec{d}(\bvec{u},\bvec{v})} = \inprod{\bvec{w}_{T}}{\bvec{B}(\bvec{v} - \bvec{u})_{T}}$, where $\bvec{B} = \left(\bvec{I}_{\abs{T}} - \rho \bvec{C}\right),$  with $\bvec{C} = \bvec{\Phi}_{T}^t\bvec{\Phi}_{T}\bvec{P}^{-1}$, where $\bvec{P}^{-1} = \int_0^1 \bvec{M}_{T}^{-1}(\bvec{u} + \theta (\bvec{v} - \bvec{u})d\theta $, which follows from the following observation: \begin{align*}
\lefteqn{\bvec{\Psi}^{-1}(\bvec{v}) - \bvec{\Psi}^{-1}(\bvec{u})} & &\\
\ & = \int_0^1 \bvec{D\Psi}^{-1}(\bvec{u} + \theta (\bvec{v} - \bvec{u}))(\bvec{v} - \bvec{u})d\theta\\
\ & \stackrel{\mbox{\footnotesize{Lemma~\ref{lem:jacobians-of-psi-and-psi-inverse} }}}{=} \int_0^1 \bvec{M}^{-1}(\bvec{u} + \theta (\bvec{v} - \bvec{u}))(\bvec{v} - \bvec{u})d\theta.
\end{align*}
Thus, by Cauchy-Schwarz, $\inprod{\bvec{w}}{\bvec{d}(\bvec{u},\bvec{v})} \le \lambda_{\max}\left(\bvec{B}\right)\opnorm{\bvec{w}_{T}}_2\opnorm{(\bvec{v} - \bvec{u})_{T}}_2$, where $\lambda_{\max}(\bvec{B}) = \max\left\{\lambda_1,\ \lambda_2\right\}$, where, $\lambda_1 = \rho\lambda_{\max}\left(\bvec{C}\right) - 1$, and $\lambda_2 = 1 - \rho\lambda_{\min}\left(\bvec{C}\right)$. Now, we claim that under the assumption $\mu(1+\delta_{2K}) < (1-L)^2/(1 - l)$, $\lambda_2 > \lambda_1$. To prove this assume that $\lambda_2 \le \lambda_1$. Then, $2\rho\lambda_{\max}\left(\bvec{C}\right) \ge \rho\lambda_{\max}\left(\bvec{C}\right) + \rho\lambda_{\min}\left(\bvec{C}\right) \ge 2 \implies \rho\lambda_{\max}(\bvec{C})\ge 1 $. Now, observe that using Jensen's inequality, we get $\lambda_{\max}(\bvec{C})\le (1+\delta_{\abs{T}})\lambda_{\max}\left(\int_0^1 \bvec{M}_{T}^{-1}(\bvec{u} + \theta (\bvec{v} - \bvec{u})d\theta ) \right)\le (1+\delta_{\abs{T}}) \int_0^1 \lambda_{\max}\left(\bvec{M}_{T}^{-1}(\bvec{u} + \theta (\bvec{v} - \bvec{u})\right)d\theta\le (1+\delta_{\abs{T}})\int_0^1\frac{1}{1-\max_{i\in T}\gamma_i g_i''(\psi^{-1}(u_i + \theta(v_i - u_i)) )}d\theta\le (1+\delta_{\abs{T}})\int_0^1 \frac{1}{1-L}d\theta = \frac{1+\delta_{\abs{T}}}{1-L}$. Thus, $\rho \lambda_{\max}(\bvec{C})\ge 1\implies \frac{\rho(1+\delta_{\abs{T}})}{1-L} >1$, which implies $1 < \frac{1-L}{1-l}\implies L<l$ which is a contradiction. Thus, $\lambda_{\max}(\bvec{B}) = 1 - \rho\lambda_{\min}(\bvec{C})$. Now, it is easy to see that $\lambda_{\min}(\bvec{C}) = \lambda_{\min}\left(\bvec{\Phi}_{T}^t\bvec{\Phi}_{T}\bvec{P}^{-1}_{T}\right) = \lambda_{\min}\left(\bvec{P}^{-1/2} \bvec{\Phi}_{T}^t\bvec{\Phi}_{T} \bvec{P}^{-1/2}\right) \ge \lambda_{\min}(\bvec{\Phi}^t_{T}\bvec{\Phi}_{T})\lambda_{\min}(\bvec{P}^{-1})$. Moreover, using Jensen's inequality we get $\lambda_{\min}(\bvec{P}^{-1}) = \lambda_{\min}\left(\int_0^1\bvec{M}^{-1}_{T}(\bvec{u} + \theta (\bvec{v} - \bvec{u}))d\theta\right)\ge \int_0^1 \lambda_{\min}\left(\bvec{M}^{-1}_{T}(\bvec{u} + \theta (\bvec{v} - \bvec{u}))\right) d\theta\ge \int_0^1 \frac{1}{1-\min_{i\in T}\gamma_i g_i''(\psi^{-1}(u_i + \theta(v_i - u_i)))} d\theta \ge \frac{1}{1 -  \min_{u\in [\psi^{-1}(-E),\ \psi^{-1}(E)]}\gamma_i g_i''(u)}$ since $\abs{u_i+\theta(v_i - u_i)}\le (1-\theta)\abs{u}_i + \theta \abs{v_i}\le E$. Hence, $\lambda_{\max}(\bvec{B})\le 1 - \frac{\rho(1-\delta_{\abs{T}})}{1-l} = \rho'_{\abs{T}}$, and the inequality~\eqref{eq:inner-product-inequality} follows from an application of Cauchy-Schwarz inequality.

 On the other hand, if $\bvec{a} = \begin{bmatrix}
(\bvec{d}(\bvec{u},\bvec{v}))_{T_3}\\
\bvec{0}_{T_3}^C
\end{bmatrix} $, then, from the first part of this lemma we get \begin{align*}
\opnorm{\bvec{a}}_2^2 & = \inprod{\bvec{a}}{\bvec{d}(\bvec{u},\bvec{v})}\\
\ & \le \rho'_{\abs{T}}\opnorm{\bvec{a}}_2\opnorm{\bvec{v} - \bvec{u}}_2\\
\implies \opnorm{\bvec{a}}_2 & \le \rho'_{\abs{T}}\opnorm{\bvec{v} - \bvec{u}}_2.
\end{align*}

To check that $\rho'_{\abs{T}}\in (0,1)$, note that the condition $ \sup_{u\in \real}\gamma_i g_i''(u)>0$  for all $i=1,2,\cdots,\ L$, makes $\rho'_{\abs{T}}$ less than unity, and the condition $\rho <\frac{(1-L)^2}{(1+\delta_{\abs{T}})(1-l)}$ ensures that  $\rho'_{\abs{T}} = 1 - \frac{\mu(1-\delta_{\abs{T}})}{1-l} > 1 - \frac{(1-\delta_{\abs{T}})}{1+\delta_{\abs{T}}}\left(\frac{1-L}{1-l}\right)^2 > 0$.
%
\bibliography{mixed-function-htp}
\end{document}